%% file: arXiv2605.tex
\documentclass[a4paper,UKenglish,cleveref,autoref,thm-restate]{lipics-v2021}

\title{Proof Identity and Categorical Models of BV} 

\author{Matteo {Acclavio}}{FORM, University of Southern Denmark, DK \and University of Sussex, UK }{}{https://orcid.org/0000-0002-0425-2825}{}
\author{Lutz Stra\ss burger}{Inria \and LIX, Ecole Polytechnique, Institut Polytechnique Paris, France}{}{https://orcid.org/0000-0003-4661-6540}{Partially supported by the French Ministry for Europe and Foreign Affairs (MEAE), the Embassy of France in the UK, and the French Ministry of Higher Education and Research (MESR), via the PHC Sophie Germain project
``Using Formal Logic to Reduce Bias in Large Language Models''}
\author{Vladimir Zamdzhiev}{Université Paris-Saclay, CNRS, ENS Paris-Saclay, Inria, Laboratoire Méthodes Formelles, 91190, Gif-sur-Yvette, France}{}{https://orcid.org/0000-0002-9061-3921}{This work has been partially funded by the French
National Research Agency (ANR) within the framework of ``Plan France 2030'',
under the research projects EPIQ ANR-22-PETQ-0007, HQI-Acquisition
ANR-22-PNCQ-0001, and HQI-R\&D ANR-22-PNCQ-0002.}

\authorrunning{M.~Acclavio, L.~Stra\ss burger and V.~Zamdzhiev} 

\Copyright{Matteo Acclavio, Lutz Stra\ss burger and Vladimir Zamdzhiev} 



\category{} 

\relatedversion{} 

\acknowledgements{We thank the anonymous referees for their feedback which led to multiple improvements to the paper. We also thank James Hefford for discussions related to this paper.}

\nolinenumbers 

\EventEditors{Frank Pfenning}
\EventNoEds{1}
\EventLongTitle{11th International Conference on Formal Structures for Computation and Deduction (FSCD 2026)}
\EventShortTitle{FSCD 2026}
\EventAcronym{FSCD}
\EventYear{2026}
\EventDate{July 20--23, 2026}
\EventLocation{Lisbon, Portugal}
\EventLogo{}
\SeriesVolume{378}
\ArticleNo{22}
\theoremstyle{definition}
\newtheorem{notation}[theorem]{Notation}

\DeclareMathOperator{\hagertimes}{\ensuremath{\stackrel{\scriptscriptstyle h}{\otimes}}}%
\DeclareMathOperator{\projtens}{\ensuremath{\widehat{\otimes}}}%
\DeclareMathOperator{\ptimes}{\projtens}%
\DeclareMathOperator{\injtens}{\ensuremath{\widecheck{\otimes}}}%
\DeclareMathOperator{\itimes}{\injtens}%
\DeclareMathOperator{\ptimessk}{\underline{\projtens}}%

\makeatletter
\DeclareRobustCommand\widecheck[1]{{\mathpalette\@widecheck{#1}}}
\def\@widecheck#1#2{%
    \setbox\z@\hbox{\m@th$#1#2$}%
    \setbox\tw@\hbox{\m@th$#1%
       \widehat{%
          \vrule\@width\z@\@height\ht\z@
          \vrule\@height\z@\@width\wd\z@}$}%
    \dp\tw@-\ht\z@
    \@tempdima\ht\z@ \advance\@tempdima2\ht\tw@ \divide\@tempdima\thr@@
    \setbox\tw@\hbox{%
       \raise\@tempdima\hbox{\scalebox{1}[-1]{\lower\@tempdima\box
\tw@}}}%
    {\ooalign{\box\tw@ \cr \box\z@}}}
\makeatother

\newcommand{\lutz}[1]{{\color{blue}[[Lutz: #1]]}}

\usepackage{marginnote}

\usepackage[T1]{fontenc}
\usepackage{mathtools}
\usepackage{virginialake}
\usepackage{stmaryrd}
\vlnosmallleftlabels
\vlnostructuressyntax
\odframefalse
\odbackgroundtrue
\odbackgroundfirstfalse
\usepackage[lneg]{proofzilla}
\usepackage{newaf}
\usepackage{xspace}
\input{macros.tex}
\usepackage{tikz}
\usepackage{tikz-cd}
\usepackage{adjustbox}
\usepackage{physics}

\newcommand{\AFsymb}{\mathsf{AF}}
\newcommand{\AFof}[1]{\AFsymb\left(#1\right)}
\newcommand{\Yof}[1]{\mathsf{Y}\left(#1\right)}

\newcommand{\afD}[1][]{\phi_{#1}}

\def\iof#1{\mathsf{in}_{#1}}
\def\oof#1{\mathsf{out}_{#1}}
\def\i{\mathsf{i}}

\def\yanksto{\stackrel{\mathsf{Y}}{\rightsquigarrow}}

\newcommand{\vlun}{\mathbb{I}}
\renewcommand{\lunit}{\vlun}

\newcommand{\ccunit}{J} 

\newcommand{\CEto}{\stackrel{\mathsf{CE}}{\leadsto}}

\newcommand{\swap}{\sigma}
\newcommand{\id}{\ensuremath{\mathrm{id}}}%
\newcommand{\CC}{\ensuremath{\mathbf{C}}}%
\newcommand{\DD}{\ensuremath{\mathbf{D}}}%
\newcommand{\sem}[1]{{\left\llbracket #1 \right\rrbracket}}%
\newcommand{\afsem}[1]{{\llparenthesis #1 \rrparenthesis}}%
\newcommand{\Afsem}[1]{\left(\!\left|#1\right|\!\right)}

\newcommand{\FdVect}{\ensuremath{\mathbf{FdVect}}}%
\newcommand{\FdVectsk}{\ensuremath{\mathbf{FdVect_{sk}}}}%
\newcommand{\fdOS}{\ensuremath{\mathbf{FdOS}}}%
\newcommand{\FdOS}{\fdOS}%
\newcommand{\FdOSsk}{\mathbf{FdOS_{sk}}}%
\newcommand{\Caus}{{\ensuremath{\mathbf{Caus}}}}%

\newcommand\transpose{\cdot^{\mathrm{t}}}
\newcommand\catid{\mathsf{Id}}

\newcommand{\mix}{\mathsf{mix}}
\newcommand{\mixz}{\mathsf{mix_0}}

\newcommand{\catC}{\ensuremath{\mathbf{C}}}%
\newcommand{\catD}{\ensuremath{\mathbf{D}}}%

\renewcommand{\lpar}{\vlpa}

\begin{document}

\maketitle
\begin{abstract}
  \bvcats are a recent development that aims to give categorical semantics to proofs in the logic $\BV$. However, due to the absence of a coherence theorem on one side and a well-defined notion of proof identity for $\BV$ on the other side, the precise relation between \bvcats and the logic $\BV$ is still not clear. To improve on this situation, we define in this paper a notion of proof identity for $\BV$, based on the notion of atomic flows, which can be seen as a special form of string diagrams. 
  Based on this notion of proof identity, we then strengthen the existing notion of \bvcat and prove that it is sound with respect to the logic.
\end{abstract}

\newpage
\section{Introduction}

The logic $\BV$~\cite{gug:str:01,gug:SIS} is an extension of multiplicative linear logic ($\MLL$) with a self-dual non-commutative connective~$\vlse$, called \emph{seq}. It came to life because of an attempt to give a deductive proof system to \emph{pomset logic}~\cite{retore:phd,retore:99,ret:newPomset}, which is defined in terms of proof nets and shares with $\BV$ the same language for formulas. Only recently, it has been shown that the two logics are different~\cite{tito:lutz:csl22,tito:str:SIS-III}. 
Nonetheless, the logic $\BV$ has successfully been used to model sequentiality in the context of concurrency~\cite{bru:02,hor:tiu:19,hor:tiu:ama:cio:private} and causality in the context of quantum computing~\cite{blu:pan:slav:deep,blu:gug:iva:pan:str:quantum}.
On the one hand, the syntax and proof theory of $\BV$ is well-developed \cite{gug:SIS,tiu:SIS-II,gug:str:01}, and it has a wide range of applications \cite{acc:str:IBV,bru:02,hor:tiu:ama:cio:private},
but on the other hand, the semantics and possible models of $\BV$ are not so well investigated. 
This is currently changing through the increasing interest from the quantum computing community~\cite{sim:kiss:BV,simmons:phd,sim:kis:complete,qcs} and the development of the notion of \bvcat~\cite{blu:pan:slav:deep,bv-cats2}.
Some non-trivial examples of \bvcats are the \emph{higher-order causal theories} of~\cite{sim:kiss:BV} and the \emph{finite-dimensional operator spaces} of~\cite{qcs}.

A natural question to ask is whether the aforementioned \bvcats are \emph{sound} models of $\BV$.
This question is less silly than it seems, as the logic
$\BV$ and the notion of \bvcat have been investigated by essentially
disjoint communities. 
The precise relation between \bvcats and the logic $\BV$ is not known and neither is soundness.

In this paper we want to improve on this situation and bring the two communities together by proving a soundness result: whenever two proofs are equal in $\BV$, then they have the same interpretation in the categorical model.  
But for that, we need to answer another question: When are two $\BV$-proofs equal? 
One of the most widely accepted notions of proof identity is via proof normalization, i.e., cut elimination. 
However, $\BV$ does not have a sequent calculus proof system~\cite{tiu:SIS-II,tito:str:SIS-III}, but is presented in a deep inference setting. This means that cut elimination for $\BV$ is not proved by the standard methods, but via a \emph{splitting lemma}~\cite{gug:SIS,tubella:phd}, which is hard to mimic in a categorical setting.
Furthermore, even in the cut-free case, it is not clear when two $\BV$-proofs should be the same, as deep inference allows for more rule permutations than a sequent system.

For this reason, we introduce in this paper a notion of \emph{atomic flows}~\cite{gug:gun:flows,gug:gun:str:LICS10} for $\BV$, that we use to identify proofs, i.e., two $\BV$ proofs are the same if they are mapped to the same atomic flow. In that respect, atomic flows share properties of \emph{proof nets}~\cite{gir:ll,girard:96:PN} and \emph{string diagrams}~\cite{lafont:boolean,polybook,string:diagrams:RTI,string:diagrams:RTII}.

The next step, after having established a notion of proof identity for $\BV$,
is to find the right notion for a categorical model of $\BV$. However, it is not clear if the existing proposals for \bvcats~\cite{blu:pan:slav:deep,bv-cats2} are invariant with respect to cut elimination
and they also do not seem to match well with the notion of proof identity that we introduce here.
In order to provide a categorical semantics for $\BV$ that is sound in the sense discussed above, we propose a new notion of \bvcat that we call \emph{strong \bvcat}.
Indeed, a strong \bvcat is also a \bvcat in the sense of~\cite{blu:pan:slav:deep,bv-cats2} (Theorem \ref{thm:strong-bv-category}), but the converse implications are unlikely to hold, in general.
The key idea behind the definition of a strong \bvcat (Definition \ref{def:bv-model}) is the following: our atomic flows are string diagrams that are known to be the internal language of \emph{strict compact closed categories}~\cite{kelly:laplaza:80,selinger:dagger}, and since the notion of proof identity is based on atomic flows, we may simply lift the coherence properties of strict compact closed categories (which are already known \cite{kelly:laplaza:80}) to the relevant diagrams of the \bvcat in which we interpret our proofs by assuming the existence of a suitable functor.

In summary, our paper makes the following contributions:
\begin{enumerate}
  \item We introduce atomic flows for $\BV$ and use them to define proof equivalence. We show that the \emph{splitting lemma}~\cite{gug:SIS,tubella:phd}, which is at the core of the cut elimination procedure, preserves the atomic flows, and therefore cut elimination in $\BV$ corresponds to \emph{yanking} in atomic flows:
  \begin{equation}\label{eq:yankingIntro}
  \hfill
  \begin{AF}
    \naf{-1,-.5}{\nafv{.75}}
    \naf{0,1}{\nafhdx{}{}{}{}{0.75}}
    \naf{2,-.5}{\nafhux{}{}{}{}{0.75}}
    \naf{3,1}{\nafv{.75}}
  \end{AF}
  \quad\yanksto\quad
  \MAF{\nafv{1.5}}
  \qqquand
  \begin{AF}
    \naf{3,-.5}{\nafv{.75}}
    \naf{2,1}{\nafhdx{}{}{}{}{0.75}}
    \naf{0,-.5}{\nafhux{}{}{}{}{0.75}}
    \naf{-1,1}{\nafv{.75}}
  \end{AF}
  \quad\yanksto\quad
  \MAF{\nafv{1.5}}
  \hfill
  \end{equation}
  Therefore, if we impose~\eqref{eq:yankingIntro} as an equality on atomic
  flows, then atomic flows are preserved under cut elimination.
  Note that this is not trivial: whereas in linear logic yanking corresponds to cut elimination \cite{kelly:maclane:71,hughes:freestar}, this is not the case in classical logic \cite{gug:gun:str:LICS10}.

  \item 
  We introduce a notion of \emph{strong \bvcat} which has a suitable functor into a strict compact closed category as part of the definition. This gives us a sufficient amount of coherence properties which we use to prove our main semantic result: if two $\BV$ proofs are equivalent, then they have the same interpretation in the model. We also show that the strictified versions of the aforementioned examples of higher-order causal theories~\cite{sim:kiss:BV} and finite-dimensional operator spaces~\cite{qcs} form strong \bvcats in our sense.
\end{enumerate}
\subparagraph*{Outline of the paper.} In Section~\ref{sec:prelims}, we recall the logic $\BV$, then, in Section~\ref{sec:AF}, we introduce atomic flows, and in Section~\ref{sec:splitting} we prove that they are preserved under cut elimination in $\BV$. Then, in Section~\ref{sec:BVcats} we introduce our \emph{strong \bvcats} and prove their soundness for~$\BV$ in \Cref{sec:semantics}. Finally, in \Cref{sec:concrete} we discuss some concrete examples of strong \bvcats.

\begin{figure}
  \def\myskip{\qqquad}
  \centering
  $\begin{array}{c@{\myskip}c@{\myskip}c@{\myskip}c}
    &
    &
    \nodn{\lunit}{\aidr}{a\lpar \lneg a}{}
    &
    \nodn{(A\lpar B)\lseq (C\lpar D)}{\qdr}{(A\lseq C)\lpar(B\lseq D)}{}
    \\
    \smash{\nodt{\;\;A\;\;}{\feq}{B}{\text{\scriptsize if $A\feq B$}}}
    &
    \smash{\nodn{A\ltens (B\lpar C)}{\swir}{(A\ltens B)\lpar C}{}}
    &
    \\
    &&
    \nodn{a\ltens \lneg a}{\aiur}{\lunit}{}
    &
    \nodn{(A\lseq B)\ltens(C\lseq D)}{\qur}{(A\ltens C)\lseq(B\ltens D)}{}
    \\&&&\\[-1.5ex]
    \text{\defn{equivalence}} & \text{\defn{switch}} & \text{\defn{atomic interaction}} & \text{\defn{sequentiality}}
  \end{array}$
  \caption{The inference rules of the proof systems $\BV=\set{\feq,\aiD,\sw,\qD}$ and $\SBV=\BV\cup\set{\aiU,\qU}$.}
  \label{fig:BV}
\end{figure}

\section{Preliminaries on BV}\label{sec:prelims}

The set $\Fms=\set{A,B,C,\ldots}$ of formulas is generated by a countable set of \defn{atoms} $\atomset$ and a \defn{unit} $\lunit$ using the following grammar:
\begin{equation}\label{eq:formulas}
  \hfill
  A ,B \coloneqq \lunit \mid a \mid \lneg a \mid (A \lpar B) \mid (A \lseq B) \mid (A \ltens B)
  \qquad
  \text{ with } a \in \atomset
  \hfill
\end{equation}
using the binary connectives $\lpar$ (\defn{par}), $\lseq$ (\defn{seq}), and $\ltens$ (\defn{tensor}). The \defn{negation} $\lneg{(\cdot)}$ can be extended to all formulas via the following \emph{De Morgan laws}:
\begin{equation}\label{eq:lneg}
  \lneg\lunit=\lunit
  \qomma
  \lneg{(\lneg a)}=a
  \qomma
  \lneg{(A\lpar B)}=\lneg A\ltens\lneg B
  \qomma
  \lneg{(A\lseq B)}=\lneg A\lseq\lneg B
  \qomma
  \lneg{(A\ltens B)}=\lneg A\lpar\lneg B
\end{equation}
It follows that $\lneg{(\lneg A)}=A$ for all~$A$. A \defn{context} is a formula with a unique occurrence of the \defn{hole} $\conhole$ occurring at the place of an atom:
$$
\def\mymid{\;\mid\;}
\cC                 \;\coloneqq\;
\ctx                \mymid
B\lpar \cC          \mymid
\cC\lpar B          \mymid
B\lseq \cC          \mymid
\cC\lseq B          \mymid
B\ltens \cC         \mymid
\cC\ltens B
$$
Then $\cC[A]$ stands for the formula obtained from the context $\cC$ by replacing the hole $\ctx$ by the formula $A$. 
We also employ a \defn{formula equivalence} $\feq$ to be the smallest congruence relation generated by:
\begin{equation}\label{eq:feq}
  \hfill
  \begin{array}{c@{\qqquad}c@{\qqquad}r@{\qqquad}}
    (A\ltens B)\ltens C   \feq    A\ltens(B\ltens C)
  &
    A\ltens B             \feq    B\ltens A
  &
    A\ltens\lunit         \feq    A
  \\
    (A\lpar B)\lpar C     \feq    A\lpar(B\lpar C)
  &
    A\lpar B             \feq    B\lpar A
  &
    A\lpar\lunit          \feq    A
  \\
    (A\lseq B)\lseq C     \feq    A\lseq(B\lseq C)
  &&
    A\lseq\lunit          \feq    A \rlap{$\;\feq \lunit\lseq A$}
  \end{array}
  \hfill
\end{equation}
An \defn{inference rule} is a scheme $\vlupsmash{\vlinf{\rR}{}{B}{A}}$ where $\rR$ is its name, $A$ its \defn{premise} and $B$ its \defn{conclusion}. A \defn{proof system} is a set of inference rules, and the rules we use in this paper are shown in \Cref{fig:BV}. For the inference rule $\feq$, we have the side condition that $A\feq B$ via one of the equations in~\eqref{eq:feq} above. We write it with a dotted line to ease readability of longer derivations, and we also often condense several such steps into a single one. 
We call $\BV$ the proof system $\set{\feq,\aiD,\sw,\qD}$ and the proof system $\SBV=\BV\cup\set{\aiU,\qU}$.

To build \emph{derivations} we use the syntax of \emph{open deduction}~\cite{gug:gun:par:RTA10}, which allows to use the same primitives that are used to construct formulas also for constructing derivations.

\begin{definition}
   We define the set $\Deris=\set{\dD,\dDb,\ldots}$ of \defn{derivations} together with two functions $\premf,\concf\colon\Deris\to\Fms$ (called \defn{premise} and \defn{conclusion}, respectively) inductively as follows:
  \begin{itemize}
  \item For every formula $A$, we have that $\nodf{A}$ is a derivation and $\prem{\nodf{A}}=\conc{\nodf{A}}=A$. 
  \item If $\dD$ and $\dDb$ are derivations and $\odot \in \set{\vlpa, \vlse, \vlte}$, then 
    $\nodf{\noD\odot \noDb}$ is a derivation and $\prem{\nodf{\noD\odot \noDb}}=\prem{\noD}\odot\prem{\noDb}$ and $\conc{\nodf{\noD\odot \noDb}}=\conc{\noD}\odot\conc{\noDb}$. 
  \item If $\noD$ and $\noDb$ are derivations and $\vldownsmash{\vlinf{\rR}{}{B}{A}}$ is an instance of an inference rule with $\conc\noD=A$ and $\prem\noDb=B$ then $\vldownsmash{\dDc={\odn{\dD}{\rR}{\dDb}{}}}$ is a derivation and $\prem{\dDc}=\prem{\dD}$ and $\conc{\dDc}=\conc{\dDb}$. 
  \end{itemize}
\end{definition}
If $\sysX$ is a proof system, we often write $\nosmash{\odv{A}{\dD}{B}{\sysX}}$ to denote a derivation $\dD$ with premise $\prem\dD=A$ and conclusion $\conc\dD=B$ where all inference rules used in $\dD$ are in the proof system~$\sysX$.
We may denote by $\nosmash{\odr{\dD}{B}{\sysX}}$ a derivation in~$\sysX$ with premise $\vlun$, and by $\downsmash{\odN{A}{\rR}{B}{}}$ a derivation using only instances of the rule~$\rR$.
We also write $A\,\proves[\sysX] B$ \resp{$\proves[\sysX]B$} if there is a derivation in the system~$\sysX$ with conclusion~$B$ and premise~$A$ \resp{premise~$\lunit$}.

Now consider the following two rules, called \defn{axiom} and \defn{cut}.
\begin{equation}\label{eq:cut}
\hfill
  \vlinf{\iD}{}{\lneg A \vlpa A}{\vlun}
  \qquad\qquand\qquad
  \vlinf{\iU}{}{\vlun}{A \vlte \lneg A}
\hfill
\end{equation}
They are the generalized form of the two inference rules $\aiD$ and $\aiU$ respectively. 
The cut elimination in $\BV$ is the following statement, whose proof will be discussed in~\Cref{sec:splitting}.  
\begin{theorem}\label{thm:cutelim}
  Let $A$ be a formula. If\/ $\proves[\BV\cup\set{\iU}] A$ then\/ $\proves[\BV] A$.
\end{theorem}

We conclude this section by recalling some basic and easy to prove properties of $\BV$.

\begin{proposition}\label{prop:i}
  The following hold:
  \begin{enumerate}
    \item The rule $\iD$ is derivable in $\set{\feq,\aiD,\swir,\qD}=\BV$.
    \item The rule $\iU$ is derivable in $\set{\feq,\aiU,\swir,\qU}=\SBV\setminus\set{\aiD,\qD}$.
    \item Every rule $\rU$ is derivable in $\set{\feq,\iD,\iU,\sw,\rD}$.
  \end{enumerate}
\end{proposition}
\begin{proof}
  The first two statements are proved by induction on $A$, and the third statement by the derivation
  on the left of \Cref{fig:dersProofs}.
\end{proof}

\begin{figure}
  \small\footnotesize
  \centering
  $\begin{array}{c@{\qquad\vrule\qquad}c@{\qquad\vrule\qquad}c}
    \od{
      \odo{
        \odi{
          \odo{\odh{A}}{}{
            A \vlte \odnP{\vlun}{\iD}{\odnP{\lneg B}{\rD}{\lneg A}{} \vlpa B}{}
          }{}
        }{\sw}{
          \odnP{ A \vlte \lneg A}{\iU}{\vlun}{}
          \vlpa 
          B
        }{}
      }{}{
        B
      }{}
    }
  &
    \odn{\vlun}{\idr}{
      \lneg A \vlpa \odv{A}{\dD}{B}{}
    }{}
  &
    \od{\odo{\odi{\odo{\odh{A}}{}{
      A\vlte \odv{\vlun}{\dDp}{\lneg A \vlpa B}{}
    }{}}{\swir}{
      \odnP{A\vlte \lneg A}{\iur}{\vlun}{}\vlpa B
    }{}}{}{
      B
    }{}}
  \end{array}$
  \caption{Derivations to prove \Cref{prop:i}.3 and the two directions of \Cref{prop:SBV}.2}
  \label{fig:dersProofs}
\end{figure}

\begin{proposition}\label{prop:SBV}
  Let $A$ and $B$ be formulas.
  Then,
  \begin{enumerate}
    \item $\proves[\SBV] A$ iff\/ $\proves[\BV\cup\set{\iU}] A$, and
    \item $A\proves[\SBV] B$ iff\/ $\proves[\SBV] \lneg A \vlpa B$.
  \end{enumerate}
\end{proposition}
\begin{proof}
  The first statement follows immediately from \Cref{prop:i}.
  To prove left-to-right \resp{right-to-left} implication in the second statement, we assume the existence of a derivation $\dD$ with premise $A$ and conclusion $B$ \resp{a derivation $\dDp$ with premise $\vlun$ and conclusion $\lneg A \vlpa B$}  in $\SBV$, and we construct the derivations in the center \resp{right} of \Cref{fig:dersProofs}.
\end{proof}

\section{Atomic Flows}\label{sec:AF}

\emph{Atomic flows} were originally introduced in \cite{gug:gun:flows,gug:gun:str:LICS10} as a graphical formalism to study normalization properties of proofs for classical logic. 
Because of their ``graphical nature'', atomic flows can be seen as relatives of Girard's \emph{proof nets}~\cite{gir:ll,lau:reg:89} and Buss' \emph{logical flow graphs}~\cite{buss:91,carbone:99}. Atomic flows have never been considered for linear logic because of the existence of proof nets. However, for $\BV$, we do not have proof nets,\footnote{The proof nets of \emph{pomset logic}~\cite{retore:97,retore:99} do not have a suitable correctness criterion for~$\BV$~\cite{tito:lutz:csl22,tito:str:SIS-III}.} and therefore, atomic flows for $\BV$ can also be seen as a first step towards the development of proof nets for $\BV$.  

The basic idea of atomic flows is to track atom occurrences in a derivation, abstracting away from the logical structure of the formulas (i.e., the connectives and their arrangement). In that respect, they can also be seen as a special case of \emph{string diagrams}~\cite{polybook,string:diagrams:RTI,string:diagrams:RTII}, or simply \emph{two-dimensional diagrams}~\cite{lafont:boolean}.

We write $\lneg\Atms=\set{\lneg a,\lneg b,\ldots}$ for the set of negated atoms and we let $\AFTypes=\set{p,q,s,t,\ldots}$ be the set of finite strings over the set $\Atms\cup\lneg\Atms$ of \defn{literals}. 
We write $\emptystring$ for the \defn{empty string} and $p\scomp q$ for \defn{string concatenation}.

\begin{figure}[t]
  \def\myha{1}
  \def\myhb{.75}
  \def\myhc{-1.5ex}
  \def\myequ{\;=\;}
  $$
  \begin{array}{c|c}
    \scalebox{.8}{$
      \begin{array}{c}
        \begin{tikzpicture}[nafZ]
          \nafvu{\myha}{}{a}
        \end{tikzpicture}
        \\\\[\myhc]
        \begin{tikzpicture}[nafC]
          \nafljl{a}{}{}{a}1{\myha}
          \nafljr{b}{}{}{b}1{\myha}
        \end{tikzpicture}
        \\\\[\myhc]
        \begin{tikzpicture}[nafC]
          \nafhdx{a}{}{}{\lneg a}{\myhb}
        \end{tikzpicture}
        \\\\[\myhc]
        \begin{tikzpicture}[nafC]
          \nafhux{a}{}{}{\lneg a}{\myhb}
        \end{tikzpicture}    
      \end{array}
      $}
  \quad\strut&\quad
    \scalebox{.9}{$
    \begin{array}{c@{\quad\qquad}c@{\quad\qquad}c}
      \begin{tikzpicture}[nafC]
        \naf{0,0}{\nafex11}
        \naf{0,2}{\nafex11}
      \end{tikzpicture}
    \myequ
      \begin{tikzpicture}[nafC]
        \naf{0,0}{\nafv2}
        \naf{2,0}{\nafv2}
      \end{tikzpicture}
    &
      \begin{tikzpicture}[nafC]
        \naf{0,0}{\nafhdx{\lneg a}{}{}{a}{.5}}
        \naf{0,-2}{\nafex11}
      \end{tikzpicture}
    \myequ
      \begin{tikzpicture}[nafC]
        \naf{0,0}{\nafhdx{a}{}{}{\lneg a}{1.5}}
      \end{tikzpicture}
    &
      \begin{tikzpicture}[nafC]
        \naf{0,0}{\nafhdx{a}{}{}{\lneg a}{.5}}
        \naf{-1,-2}{\nafv1}
        \naf{2,-2}{\nafex11}
        \naf{3,0}{\nafv1}
      \end{tikzpicture}
    \myequ
      \begin{tikzpicture}[nafC]
        \naf{3,0}{\nafhdx{a}{}{}{\lneg a}{.5}}
        \naf{4,-2}{\nafv1}
        \naf{1,-2}{\nafex11}
        \naf{0,0}{\nafv1}
      \end{tikzpicture}
    \\\\[-5pt]
      \begin{tikzpicture}[nafC]
        \naf{0,0}{\nafex11}
        \naf{3,0}{\nafv1}
        \naf{-1,-2}{\nafv1}
        \naf{2,-2}{\nafex11}
        \naf{0,-4}{\nafex11}
        \naf{3,-4}{\nafv1}
      \end{tikzpicture}
    \myequ
      \begin{tikzpicture}[nafC]
        \naf{3,0}{\nafex11}
        \naf{0,0}{\nafv1}
        \naf{4,-2}{\nafv1}
        \naf{1,-2}{\nafex11}
        \naf{3,-4}{\nafex11}
        \naf{0,-4}{\nafv1}
      \end{tikzpicture}
    &
      \begin{tikzpicture}[nafC]
        \naf{0,0}{\nafhux{\lneg a}{}{}{a}{.5}}
        \naf{0,2}{\nafex11}
      \end{tikzpicture}
    \myequ
      \begin{tikzpicture}[nafC]
        \naf{0,0}{\nafhux{a}{}{}{\lneg a}{1.5}}
      \end{tikzpicture}
    &
      \begin{tikzpicture}[nafC]
        \naf{0,0}{\nafhux{a}{}{}{\lneg a}{.5}}
        \naf{-1,2}{\nafv1}
        \naf{2,2}{\nafex11}
        \naf{3,0}{\nafv1}
      \end{tikzpicture}
    \myequ
      \begin{tikzpicture}[nafC]
        \naf{3,0}{\nafhux{a}{}{}{\lneg a}{.5}}
        \naf{4,2}{\nafv1}
        \naf{1,2}{\nafex11}
        \naf{0,0}{\nafv1}
      \end{tikzpicture}
    \end{array}
    $}
  \end{array}
  $$
  \caption{Generators (left) and equivalences (right) for atomic flows, with $a$ and $b$ literals.}
  \label{fig:af}
\end{figure}
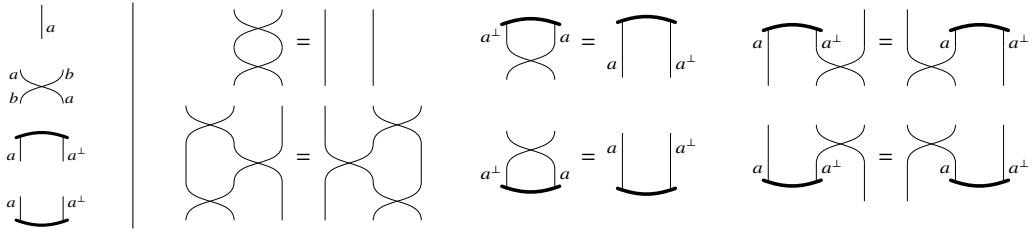

\begin{definition}
  An \defn{atomic flow} $\afD\colon p\to q$ is a two-dimensional diagram written as~~
  \upsmash{\raise1ex\hbox{\scalebox{.95}{%
  \begin{tikzpicture}[nafZ]
    \naf{0,0}{\naffplainbox{\phi}{a_1}{a_n}{b_1}{b_m}}
  \end{tikzpicture}
  }}}
  where $p=\iof{\afD}=a_1\scomp\cdots\scomp a_n$ is its \defn{input} and $q=\oof{\afD}=b_1\scomp\cdots\scomp b_m$ its \defn{output}, with $a_1,\ldots,a_n,b_1,\ldots,b_m\in\Atms\cup\lneg\Atms$.

  The set of atomic flows is generated from the basic flows shown on the left of \Cref{fig:af} 
  (the identity flow, the simple crossing, and the cap and the cups making interact dual literals) 
  via following two operations:
  \begin{itemize}
    \item \defn{vertical composition:} if $\phi\colon p\to q$ and $\psi\colon q\to s$ are atomic flows, then so is $\theta=\phi\vcomp\psi\colon p\to s$, depicted on the left of \Cref{eq:afcomp} below, and
  
    \item \defn{horizontal composition:} if $\phi\colon p\to q$ and $\psi\colon s\to t$ are atomic flows, then so is $\theta=\phi\hcomp\psi\colon p\scomp s\to q\scomp t$, depicted on the right of \Cref{eq:afcomp} below,
    \begin{equation}\label{eq:afcomp}
      \phi\vcomp\psi\;=\;
      \begin{tikzpicture}[nafC]
        \nafscale{.9}{.9}{
          \naf{0,0}{\naffplainbox{\phi}{}{}{}{}}
          \naf{0,-3}{\naffplainbox{\psi}{}{}{}{}}
          }
      \end{tikzpicture}
      \quad\text{if $\oof{\phi}=\iof{\psi}$}
    \qqquand
      \phi\hcomp\psi\;=\;
      \begin{tikzpicture}[nafC]
        \nafscale{.9}{.9}{
          \naf{0,0}{\naffplainbox{\phi}{}{}{}{}}
          \naf{3.5,0}{\naffplainbox{\psi}{}{}{}{}}
          }
      \end{tikzpicture}
    \end{equation}
  \end{itemize}
  quotiented by the equivalence relation generated by $(\phi\vcomp\phi')\hcomp(\psi\vcomp\psi')=(\phi\hcomp\psi)\vcomp(\phi'\hcomp\psi')$ and the equations shown on the right of \Cref{fig:af}.
  We denote the \defn{empty atomic flow} by $\emptyflow\colon\emptystring\to\emptystring$.
\end{definition}

We are now going to show how $\SBV$-derivations are translated into atomic flows. For this we associate to each formula $A$ it string $\stringof{A}$ of atom occurrences, which is formally defined as follows:
\begin{equation}
  \stringof{\vlun}=\emptystring
  \qquad
  \stringof{a}=a
  \qquad
  \stringof{\lneg a}=\lneg a
  \qquad
  \stringof{A\vlpa B}=\stringof{A\vlse B}=\stringof{A\vlte B}=\stringof{A}\stringcat\stringof{B}
\end{equation}
For each formula $A$, we may denote the \defn{identity flow} $\idF[A]\colon\stringof A\to\stringof A$ as $\vert\cdots\vert$ or as \raise.6ex\hbox{$\MAF{\nafVd{.65}{}{A}}$}~. 
Moreover, we may also use the abbreviations in \Cref{fig:afAbb} when writing atomic flows.

\begin{figure}[t]
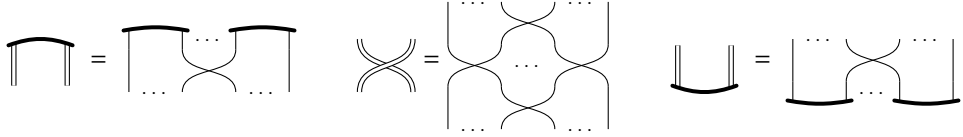

  \centering
  $\begin{array}{c}
    \MAF{\nafHd{}{}{}{}}
    =
    \begin{AF}
      \naf{-2,2}{\nafhdshort{}{}{}{}}
      \naf{2,2}{\nafhdshort{}{}{}{}}
      \naflabel{0,1.5}{\cdots}
      \naf{0,0}{\nafex11}
      \naf{-3,0}{\nafv1}
      \naf{3,0}{\nafv1}
      \naflabel{-2,-1}{\cdots}
      \naflabel{2,-1}{\cdots}
    \end{AF}
  \qquad
    \MAF{\nafEXl1{1.25}}
    =
    \begin{AF}
      \naflabel{-2,3}{\cdots}
      \naflabel{2,3}{\cdots}
      \naf{0,2}{\nafex11}
      \naf{-3,2}{\nafv1}
      \naf{3,2}{\nafv1}
      \naf{-2,0}{\nafex11}
      \naflabel{0,0}{\cdots}
      \naf{2,0}{\nafex11}
      \naf{0,-2}{\nafex11}
      \naf{-3,-2}{\nafv1}
      \naf{3,-2}{\nafv1}
      \naflabel{-2,-3}{\cdots}
      \naflabel{2,-3}{\cdots}
    \end{AF}
  \qquad
    \MAF{\nafHu{}{}{}{}}
    =
    \begin{AF}
      \naf{-2,-2}{\nafhushort{}{}{}{}}
      \naf{2,-2}{\nafhushort{}{}{}{}}
      \naflabel{0,-1.5}{\cdots}
      \naf{0,0}{\nafex11}
      \naf{-3,0}{\nafv1}
      \naf{3,0}{\nafv1}
      \naflabel{-2,1}{\cdots}
      \naflabel{2,1}{\cdots}
    \end{AF}
  \\[-10pt]
  \end{array}%
  $

  \caption{Abbreviations for atomic flows.}
  \label{fig:afAbb}
\end{figure}

We associate to each inference rule in $\SBV$ an atomic flow as shown below:
\def\hgt{.75}
\begin{equation}\label{eq:basicflows}
  \begin{array}{r@{\;=\quad}l@{\qquad\qquad}r@{\;=\;}l}
    \AFof{\aiD}
    &
    \begin{AFz}
      \naf{0,1}{\nafhdx{a}{}{}{\lneg a}{.5}}
    \end{AFz}
    &
    \AFof{\swir}
    &
    \begin{AFz}
      \naf{0,.5}{\nafVu{1}{}{A}}
      \naf{1.5,.5}{\nafVu{1}{}{B}}
      \naf{3,.5}{\nafVu{1}{}{C}}
    \end{AFz}
    \\[1ex]
    \AFof{\aiU}
    &
    \begin{AFz}
      \naf{0,0}{\nafhux{a}{}{}{\lneg a}{.5}}
    \end{AFz}
    &
    \AFof{\qD}=
    \AFof{\qU}
    &
    \begin{AFz}
      \naf{0,.5}{\nafVu{1}{}{A}}
      \naf{3,.5}{\naflJl{}{}{}{B}{1}{1}\naflJr{}{~C}{}{}{1}{1}}
      \naf{6,.5}{\nafVu{1}{}{D}}
    \end{AFz}
  \end{array}
\end{equation}
where $\AFof{\feq}$ is 
$\begin{AFz}
    \naf{0,.5}{\nafVd{\hgt}{}{A}}
    \naf{1.5,.5}{\nafVd{\hgt}{}{B}}
    \naf{3,.5}{\nafVd{\hgt}{}{C}}
\end{AFz}$ , or
$\begin{AFz}\naf{0,.5}{\naflJl{A}{}{}{}{1}{\hgt}\naflJr{}{}{}{B}{1}{\hgt}}\end{AFz}$ , 
or
$\begin{AFz}\naf{0,.5}{\nafVd{\hgt}{}{A}}\end{AFz}$~,
for the first, second, and last columns in~\eqref{eq:feq}, respectively.

We associate to each derivation $\Phi$ in $\SBV$ its atomic flow $\AFof\Phi$ inductively as follows:
\begin{itemize}
  \item For a formula $A$, we have $\AFof{A}=\idF[A]$, in particular, $\AFof{\vlun}=\emptyflow$;
  \item for $\odot\in\set{\vlpa,\vlse,\vlte}$, we define $\AFof{\Phi\odot\Psi}=\AFof{\Phi}\hcomp\AFof{\Psi}$.
  \item For $\rR\in\SBV$, and $\dDc=\vldownsmash{\odn{\dD}{\rR}{\dDb}{}}$ then $\AFof{\dDc}=\AFof{\dD}\vcomp\AFof{\rR}\vcomp\AFof{\dDb}$, where $\AFof{\rR}$ is defined as above.
\end{itemize}

\begin{proposition}\label{prop:AF-ID}
  The derivations constructed in the proof of \Cref{prop:i}.1 and \Cref{prop:i}.2 have the following atomic flows:
  \begin{equation}\label{eq:AF-i}
    \hfill
    \AFof{\odn{\vlun}{\iD}{A\vlpa \lneg A}{}}\; =\quad  \MAF{\nafHd{A}{}{}{\lneg A}}
    \hskip4em
    \AFof{\odn{A \vlte \lneg A}{\iU}{\vlun}{}}\; =\quad  \MAF{\nafHu{A}{}{}{\lneg A}}
    \hfill
  \end{equation}
\end{proposition}
\begin{proof}
  By the same induction as in the proof of \Cref{prop:i}.
\end{proof}

\begin{definition}
  We call \defn{yanking} the rewrite relation $\yanksto$ on atomic flows that is generated by the following two rewrite rules:
  \begin{equation}\label{eq:yanking}
  \hfill
  \begin{AF}
    \naf{-1,-.5}{\nafv{.75}}
    \naf{0,1}{\nafhdx{}{}{}{}{0.75}}
    \naf{2,-.5}{\nafhux{}{}{}{}{0.75}}
    \naf{3,1}{\nafv{.75}}
  \end{AF}
  \quad\yanksto\quad
  \MAF{\nafv{1.5}}
  \qquand
  \begin{AF}
    \naf{3,-.5}{\nafv{.75}}
    \naf{2,1}{\nafhdx{}{}{}{}{0.75}}
    \naf{0,-.5}{\nafhux{}{}{}{}{0.75}}
    \naf{-1,1}{\nafv{.75}}
  \end{AF}
  \quad\yanksto\quad
  \MAF{\nafv{1.5}}
  \hfill
  \end{equation}
\end{definition}

\begin{proposition}
  Yanking is terminating and confluent.
\end{proposition}

\begin{proof}
  Termination is immediate because the atomic flow gets smaller at each step, and confluence follows because the only critical pairs are given by the two ways of reducing
  $\begin{AF}
      \naf{0,1}{\nafhdshort{}{}{}{}}
      \naf{2,0}{\nafhushort{}{}{}{}}
      \naf{4,1}{\nafhdshort{}{}{}{}}
    \end{AF}
  $
  to
  $\begin{AF}
     \naf{0,1}{\nafhdshort{}{}{}{}}
  \end{AF}
  $
  and the two ways of reducing
  $\begin{AF}
      \naf{0,-1}{\nafhushort{}{}{}{}}
      \naf{2,0}{\nafhdshort{}{}{}{}}
      \naf{4,-1}{\nafhushort{}{}{}{}}
    \end{AF}
  $
  to
  $\begin{AF}
    \naf{0,1}{\nafhushort{}{}{}{}}
  \end{AF}
  $
  .
\end{proof}

This means that each atomic flow $\phi$ has a unique normal form under $\yanksto$, and we denote this normal form by $\Yof\phi$  and call it the \defn{yanking of $\phi$}.

\section{Cut Elimination is Yanking. Also in BV.}\label{sec:splitting}

Cut elimination in a deep inference system usually means that the \emph{up-fragment} (i.e., the rules with the $\uparrow$ in the name) can be eliminated. Propositions~\ref{prop:i} and~\ref{prop:SBV} above say that $\BV$ also follows this pattern. Cut elimination in $\BV$ is proved via a \emph{splitting lemma} that allows to decompose a derivation of the premise of an up-rule into a derivation of the conclusion. This \emph{splitting} (proved for $\BV$ in~\cite{gug:SIS}; see also~\cite{SIS-V,hor:tiu:ama:cio:private} for alternative presentations) is a global property of a derivation and provides, \emph{a priori}, no information about the atomic flows. But in order to establish the relation between cut elimination and yanking in $\BV$, which is a local rewrite rule, we need to strengthen the splitting lemma to also talk about the atomic flows.

\begin{restatable}[Splitting, improved formulation]{lemma}{lemSplit}\label{lem:splitting}
  Let $A$, $B$, and $K$ be formulas, and let $a$ be an atom.
  \begin{enumerate}
  \item If there is a derivation $\smash{\odr{\dD}{(A\vlte B)\vlpa K}{\BV}}$, then there are formulas $K_A$ and $K_B$ and derivations
    $$
    \odv{K_A\vlpa K_B}{\dD[K]}{K}{\BV}
    \text{ and }
    \odr{\dD[A]}{A\vlpa K_A}{\BV}
    \text{ and }
    \odr{\dD[B]}{B\vlpa K_B}{\BV}
    \text{, s.t. }
    \AFof{
      \odN{
        \odrP{\dD[A]}{A\vlpa K_A}{} \vlte \odrP{\dD[B]}{B\vlpa K_B}{}
        }{\swir}{
          (A\vlte B)\vlpa\odvP{K_A\vlpa K_B}{\dD[K]}{K}{}
        }{}
    }=\AFof{\dD}$$
  \item If there is a derivation $\smash{\odr{\dD}{(A\vlse B)\vlpa K}{\BV}}$, then there are formulas $K_A$ and $K_B$ and derivations
    $$
    \odv{K_A\vlse K_B}{\dD[K]}{K}{\BV}
    \text{ and }
    \odr{\dD[A]}{A\vlpa K_A}{\BV}
    \text{ and }
    \odr{\dD[B]}{B\vlpa K_B}{\BV}
    \text{, s.t. }
    \AFof{
      \odn{
        \odrP{\dD[A]}{A\vlpa K_A}{} \vlse \odrP{\dD[B]}{B\vlpa K_B}{}
        }{\qdr}{
          (A\vlse B)\vlpa\odvP{K_A\vlse K_B}{\dD[K]}{K}{}}
        {}
    }
    =\AFof{\dD}
    $$
  \item If there is a derivation $\smash{\odr{\dD}{a\vlpa K}{\BV}}$, then there is a formula $K$ and a derivation 
    $$
    \odv{\lneg a}{\dD[a]}{K}{\BV}
    \text{, s.t. }
    \AFof{
      \odn{\vlun}{\aidr}{a\vlpa\odvP{\lneg a}{\dD[a]}{K}{}}{}
    }=\AFof{\dD}$$
  \end{enumerate}
\end{restatable}
\begin{proof}
  The proof is the same as for the standard splitting lemma for $\BV$ \cite{gug:SIS,SIS-IV}, which proceeds by induction on the size of $\Phi$ and a case analysis on the bottommost rule instance in~$\Phi$. The only difference here is the additional observation that each transformation preserves the atomic flow. More details are in Appendix~\ref{app:1}.
\end{proof}

In the next step, the context $\conhole\vlpa K$ of the splitting lemma is generalized to an arbitrary context $\cC$.
\begin{restatable}[Context Reduction, improved formulation]{lemma}{lemCR}\label{lem:context}
  If there is a derivation $\smash{\odr{\dD}{\cC[A]}{\BV}}$, then there is a formula $K$ 
  and derivations
  $$
  \odr{\dD[A]}{A\vlpa K}{\BV}
  \quand
  \odv{X\vlpa K}{\dD[{\cC[X]}]}{\cC[X]}{\BV}
  \text{ for every formula $X$, such that\/ }
  \AFof{
    \od{\odd{\odp{\dD[A]}{A\vlpa K}{\BV}}{\dD[{\cC[A]}]}{\cC[A]}{\BV}}
  }=\AFof{\dD}
  $$
\end{restatable}
\begin{proof}
  As in the previous lemma, the proof is a straightforward adaptation of the standard context reduction lemma for $\BV$ \cite{gug:SIS,SIS-IV}, which proceeds by induction on the structure of $\cC$, by repeatedly applying splitting, until the base case $\conhole\vlpa K$ is reached. It follows immediately that each transformation preserves the atomic flow.
\end{proof}

These two lemmas are enough to eliminate the rules $\qU$ and $\aiU$ from a derivation. Again, this has already been shown in \cite{gug:SIS,SIS-IV}, and we only need to observe what happens to the atomic flows under this transformation.

\begin{restatable}{lemma}{lemQuElim}\label{lem:qur}
  If there is a derivation 
  $\dD= {\odr{\dD'}{\cC[\odn{(A\vlse B)\vlte(C\vlse D)}{\qur}{(A\vlte C)\vlse(B\vlte D)}{}]}{\BV}}$ 
  then there is a derivation 
  $\upsmash{\odr{\dDast}{\cC[(A\vlte C)\vlse(B\vlte D)]}{\BV}}$
  such that $\AFof{\dDast}=\AFof{\dD}$.
\end{restatable}
\begin{proof}
  By the previous two lemmas, we have derivations
  $$
  \odv{((A\vlte C)\vlse(B\vlte D))\vlpa(K_A\vlse K_B)\vlpa(K_C\vlse K_D)}
      {\dD[{\cC[(A\vlte C)\vlse(B\vlte D)]}]}{\cC[(A\vlte C)\vlse(B\vlte D)]}{\BV}
  \qomma
  \odr{\dD[A]}{A\vlpa K_A}{\BV}
  \qomma
  \odr{\dD[B]}{B\vlpa K_B}{\BV}
  \qomma
  \odr{\dD[C]}{C\vlpa K_C}{\BV}
  \qomma
  \odr{\dD[D]}{D\vlpa K_D}{\BV}
  $$
  Which can be put together to get a derivation $\dDast$. 
  The fact that $\AFof{\dDast}=\AFof{\dD}$ follows from the fact that atomic flows are equivalent modulo the identities in \Cref{fig:af}. More details are in Appendix~\ref{app:1}.
\end{proof}

\begin{lemma}\label{lem:aiur}
  If there is a derivation 
  $\dD=\vldownsmash{\odr{\dD'}{\cC[\odn{a\vlte \lneg a}{\aiur}{\vlun}{}]}{\BV}}$ 
  then there is a derivation 
  $\vldownsmash{\odr{\dDast}{\cC[\vlun]}{\BV}}$ 
  such that 
  $\AFof{\dD}\yanksto\AFof{\dDast}$.
\end{lemma}
\begin{proof}
  By context reduction and splitting, we get two formulas $K_a$ and $K_{\lneg a}$ and derivations
  $$
  \odv{(a\ltens \lneg a)\lpar K_a\lpar K_{\lneg a}}{\dD[{\cC[a\ltens\lneg a]}]}{\cC[a\ltens\lneg a]}{\BV}
  \quand
  \odv{\vlun\lpar K_a\lpar K_{\lneg a}}{\dD[{\cC[\vlun]}]}{\cC[\vlun]}{\BV}
  \quand
  \odv{\lneg a}{\dD[{a}]}{K_{a}}{\BV}
  \quand
  \odv{a}{\dD[{\lneg a}]}{K_{\lneg a}}{\BV}
  $$
  We can then build the derivations $\dDp$ and $\dDast$ shown below.
  $$
    \dDp
  =
    \od{
      \odd{
        \odI{
          \odh{
            \odnP{\vlun}{\aidr}{a\vlpa \oD[a]}{}
            \vlte
            \odnP{\vlun}{\aidr}{\lneg a \vlpa \oD[{\lneg a}]}{}
          }
        }{\swir,\feq}{
          (a\vlte \lneg a) \vlpa K_a\vlpa K_{\lneg a}
        }{}
      }{
        \dD[{\cC[a\vlte \lneg a]}]
      }{
        \cC[\odn{a\vlte \lneg a}{\aiur}{\vlun}{}]
      }{}
    }
  \qquad
    \dDast
  =
    \od{
      \odo{
        \odi{
          \odh{\vlun}}{
          \aiD}{
          \odvP{\lneg a}{\dD[{a}]}{K_{a}}{\BV}
          \vlpa 
          \odvP{a}{\dD[{\lneg a}]}{K_{\lneg a}}{\BV}
        }{}}{
      \feq}{
        \odv{\vlun \vlpa K_a\lpar K_{\lneg a}}{
          \dDp[{\cC[\vlun]}]
        }{
          \cC[\vlun]
        }{}}{
      }
    }
  $$
  We have $\AFof{\dD}\;=\;\AFof{\dDp}$ and 
  $
  \AFof\dDast
  =
  \begin{AF}
    \naf{4.5,2}{\nafhdshort[2.5]{}{}{}{}}
    \naf{2,0}{\nafdiagram{\AFof{\dD[K_a]}}{}{,.,}21}
    \naf{7,0}{\nafdiagram{\AFof{\dD[K_{\lneg a}]}}{}{,.,}21}
    \naf{4.5,-2.5}{\nafdiagram{\AFof{\dD[{\cC[\vlun]}]}}{}{,.,}{4.5}1}
  \end{AF}
  $
  by \Cref{lem:splitting,lem:context}, and
  $$
  \AFof{
    \odv{
      (a\vlte \lneg a) \vlpa K_a\vlpa K_{\lneg a}
    }{
      \dD[{\cC[a\vlte \lneg a]}]
    }{
      \cC[\odn{a\vlte \lneg a}{\aiur}{\vlun}{}]
    }{}
  }
  \;=\;
  \AFof{
    \od{
      \odd{
        \odh{
          \odnP{a\vlte \lneg a}{\aiur}{\vlun}{} \vlpa K_a\vlpa K_{\lneg a}
        }
      }{
        \dD[{\cC[\vlun]}]
      }{
        \cC[\vlun]
      }{}
    }
  }
  $$
  by rewriting via the equalities in \Cref{fig:af}. Therefore 
  $\AFof\dD$ is as shown on the left below. This rewrites with a single $\yanksto$ step, into the atomic flow on the right below, which is equal to $\AFof{\dDast}$.  
  $$\adjustbox{max width=.9\textwidth}{$
  \AFof{\Phi}
  =
  \MAF{
    \naf{2,1.8}{\nafdiagram{\AFof{\dD[{\cC[\vlun]}]}}{}{,.,}{8}{1.2}}
    \naf{0,4.5}{\nafdiagram{\AFof{\dD[a]}}{}{,.,}{3}{1}}
    \naf{7,4.5}{\nafdiagram{\AFof{\dD[\lneg a]}}{}{,.,}{3}{1}}
    \naf{0,9}{\nafhdx{a}{}{}{\lneg a}{.5}}
    \naf{1,7}{\nafv{1.5}}
    \naf{7,9}{\nafhdx{\lneg a}{}{}a{.5}}
    \naf{8,7}{\nafv{1.5}}
    \naf{-5,4}{\nafhux{}{}{}{}{.5}}
    \naf{1,6.5}{\nafjr5{1.5}}
    \naf{-3.5,6.5}{\nafjr{2.5}{1.5}}
  }
  \quad\yanksto\quad
  \MAF{
    \naf{2,1.8}{\nafdiagram{\AFof{\dD[{\cC[\vlun]}]}}{}{,.,}{8}{1.2}}
    \naf{0,4.5}{\nafdiagram{\AFof{\dD[a]}}{}{,.,}{3}{1}}
    \naf{7,4.5}{\nafdiagram{\AFof{\dD[{\lneg a}]}}{}{,.,}{3}{1}}
    \naf{7,9}{\nafhdx{\lneg a}{}{}a{.5}}
    \naf{8,7}{\nafv{1.5}}
    \naf{3.5,6.75}{\nafjr{2.5}{1.25}}
  }
  =\AFof{\dDast}
  $}$$
\end{proof}

Now \Cref{thm:cutelim} follows immediately from the previous two lemmas and \Cref{prop:SBV}.
However, with atomic flows, we have a stronger result.

\begin{notation}
  Let $\dD$ be an $\SBV$ derivation with premise $\vlun$. We write $\dD\CEto\dDp$ if $\dDp$ is a $\BV$ derivation that is obtained from $\dD$ by eliminating the $\qU$- and $\aiU$-rule instances according to the proofs of Lemmas~\ref{lem:qur} and~\ref{lem:aiur}.
  \end{notation}

We can now state and prove  our main result.
\begin{theorem}\label{thm:YAF}
  If $\dD\CEto\dDp$ then $\AFof{\dDp}=\Yof{\AFof{\dD}}$.
\end{theorem}
\begin{proof}
  By Lemmas~\ref{lem:qur} and~\ref{lem:aiur}, we have $\AFof{\dD}\yanksto\AFof{\dDp}$. Since $\dDp$ has no $\aiU$-instance (and no $\iU$-instance), we have that $\AFof{\dDp}$ has no yanking redex. Hence $\AFof{\dDp}=\Yof{\AFof{\dD}}$.
\end{proof}

\begin{corollary}\label{prop:SBV-AF}
  Let $\dD$ be a derivation of $A$ in $\BV\cup\set{\iU}$.
  Then, there is a derivation $\dDp$ of $A$ in $\BV$ such that $\AFof{\dDp}=\Yof{\AFof{\dD}}$.
\end{corollary}
\begin{proof}
  By \Cref{prop:i}, we can construct a derivation $\dDp$ by replacing each $\iU$-instance in $\dD$ by a derivation in $\set{\feq,\aiU,\swir,\qU}$.
  By \Cref{prop:AF-ID} the atomic flow of $\dDp$ is the same as the one of $\dD$.
  Finally, by \Cref{thm:YAF}, we can eliminate the up-rules from $\dDp$ obtaining a derivation $\dDs$ in $\BV$ such that $\AFof{\dDs}=\Yof{\AFof{\dDp}}=\Yof{\AFof{\dD}}$.
\end{proof}

We can now formally define proof identity in $\BV$.

\begin{definition}
  Let $\Phi$ and $\Psi$ be two derivations of $A$ in $\SBV$.
  We say that $\Phi$ and $\Psi$ are \defn{equivalent} (denoted $\Phi \preq \Psi$) if there are derivations $\Phi'$ and $\Psi'$ such that $\Phi \CEto \Phi'$ and\/ $\Psi \CEto \Psi'$ and\/ $\AFof{\Phi'} = \AFof{\Psi'}$~.
\end{definition}

\section{What is a BV-Category?}\label{sec:BVcats}

It is commonly expected from a categorical model that it identifies proofs that
are considered equal by the logic. A typical example are $*$-autonomous
categories that identify proofs of multiplicative linear logic ($\MLL$) up to
rule permutations in the sequent calculus~\cite{blute:93,lam:str:06:freestar,hughes:freestar}.

Recall that if $\tuple{\CC, \otimes, \lone,\lneg{\cdot}}$ is a $*$-autonomous category, then
we can define another symmetric monoidal structure $A \vlpa B \triangleq \lneg{(\lneg{A} \otimes \lneg B)}$ with unit $\lbot \triangleq \lneg\lone$ by duality.
As $\BV$ is a conservative extension of $\MLL$ (with $\mix$ and
$\mixz$)\footnote{$\mix$ is $\lbot\limp\lone$ and $\mixz$ is $\lone\limp\lbot$. See \cite{fle:ret:mix,abramsky:jagadeesan:94}.}, it immediately follows that any categorical
model of $\BV$ should be a $*$-autonomous category that is also isomix (this allows us to identify the two units $\lone$ and $\lbot$). Following the notation of the previous sections, we denote
this unit by~$\lunit$. Furthermore, since $\BV$ has an additional
non-commutative binary connective $\vlse$ (with unit $\lunit$), we also use an
additional monoidal structure $\tuple{\vlse,\lunit}$.

The first proposal for what a \bvcat should be was made in
\cite{blu:pan:slav:deep} followed by another one in \cite{bv-cats2}. We invite
the reader to consult Appendix \ref{app:BVcat} for the definition of a
\emph{\bvcat with negation} in the sense of \cite{blu:pan:slav:deep}.
The definition of a \bvcat with negation in \cite{blu:pan:slav:deep}
is equivalent to the definition of a \bvcat in \cite{bv-cats2}\footnote{James Hefford, personal communication.} and henceforth
we often refer to both of these (equivalent) concepts as \emph{\bvcat} for simplicity.
The presentation of \bvcats in \cite{blu:pan:slav:deep} is based on linearly distributive categories whereas
the presentation in \cite{bv-cats2} is based on $*$-autonomous categories.

Unfortunately, soundness has not been demonstrated for the two aforementioned
proposals for \bvcats. Furthermore, we conjecture that the notion of
proof identity, based on atomic flows, that we consider in this paper does not
match well with these proposals. Note that atomic flows in the logic $\BV$ correspond
to the string diagrams of a strict compact closed category (see previous
section). The \bvcats of
\cite{blu:pan:slav:deep,bv-cats2} are not built around compact closed
categories and we conjecture that, in general, arbitrary \bvcats in the sense of \cite{blu:pan:slav:deep,bv-cats2} cannot be embedded
in a compact closed category in the way we need~(Definition~\ref{def:bv-model}).

Because of this, we believe that our notion of proof
equivalence does not match well with the \bvcats of
\cite{blu:pan:slav:deep,bv-cats2}. This is our justification for introducing a
new notion of \bvcat (Definition \ref{def:bv-model}) for which we prove
soundness (Theorem \ref{thm:soundness}) and which behaves well with respect to
our notion of proof equivalence based on atomic flows
(\cref{cor:proof-equivalence}). We give our proposal the name \emph{strong
\bvcat}, because every strong \bvcat is also a \bvcat  in the sense of \cite{blu:pan:slav:deep,bv-cats2} (Theorem
\ref{thm:strong-bv-category}), but we conjecture that the converse implication
does not hold.

Before we formulate our main definition, recall that every compact closed category $(\DD,\boxtimes,\ccunit,\lneg{\cdot})$ is $*$-autonomous with $A \vlpa B  \cong A \boxtimes B.$
In this paper we work with \emph{strict} compact closed categories. By this we mean that:
(1) the symmetric monoidal structure is strict so that $A \boxtimes (B \boxtimes C) = (A \boxtimes B) \boxtimes C$ and $\ccunit\boxtimes A = A = A \boxtimes \ccunit$;
(2) the canonical De Morgan isomorphism is the identity so that $(A \boxtimes B)^{\perp} = A^{\perp} \boxtimes B^{\perp}$;
(3)~the~double dual isomorphism is the identity so that $A = A^{\perp\perp}.$
We also recall that a strict $*$-autonomous functor $U \colon \CC \to \DD$
is a functor that preserves all of the $*$-autonomous structure up to equality, i.e., $U(A \otimes B) = UA \otimes UB$, $U(\lneg{A}) = \lneg{(UA)}$,
$U(f \otimes g) = Uf \otimes Ug$, etc.

\begin{definition}\label{def:bv-model}
  A \textbf{strong \bvcat} is a tuple $\tuple{\CC,\vlte,\vlse,\vlun,\lneg\cdot,U}$ such that:
  \begin{enumerate}
    \item\label{it:units} $\tuple{\CC,\vlte,\vlun,\lneg\cdot}$ is a $*$-autonomous category together with an isomix isomorphism $\lneg\vlun \cong \vlun$;
    \item $\tuple{\CC,\vlse,\vlun}$ is a monoidal category;
    \item there is a natural isomorphism $\kappa_{AB} \colon \lneg{A} \vlse \lneg{B} \cong \lneg{(A \vlse B)}$;
    \item there is a natural transformation $q_{ABCD} \colon (A \vlpa B) \vlse (C \vlpa D) \to (A \vlse C) \vlpa (B \vlse D)$;
    \item\label{it:coherence} $U \colon \CC \hookrightarrow \DD$ is a faithful strict $*$-autonomous functor into a strict compact closed category $(\DD, \boxtimes,\lneg\cdot, \ccunit)$ such that:
      \begin{itemize}
        \item $U$ is strict monoidal with respect to $\vlse$;
        \item $U(\kappa_{AB}) = \id_X$, where  $X = \lneg{(UA)} \boxtimes \lneg{(UB)} = \lneg{(UA \boxtimes UB)} $ in $\DD$;
        \item the following diagram commutes for any choice of objects $A,B,C,D$:
          \begin{equation}\label{eq:q-coherence}
            \hfill
            \begin{tikzcd}
              U((A \vlpa B) \vlse (C \vlpa D)) \arrow[d, "Uq"'] \arrow[equal]{r} & UA \boxtimes UB \boxtimes UC \boxtimes UD \arrow[d, "\id \boxtimes \swap \boxtimes \id"]\\
              U((A \vlse C) \vlpa (B \vlse D)) \arrow[equal]{r} & UA \boxtimes UC \boxtimes UB \boxtimes UD
            \end{tikzcd}
          \hfill
          \end{equation}
          where $\sigma$ stands for the monoidal symmetry.
      \end{itemize}
  \end{enumerate}
\end{definition}

By using the canonical De Morgan isomorphisms in the $*$-autonomous category $\CC$, the double isomorphism, and the self-duality of $\lseq$ via the $\kappa$ natural isomorphism, one may now define a
natural transformation
\[ q'_{ABCD} \colon (A \lseq B) \ltens (C \lseq D) \xrightarrow{\cong \circ \lneg{q} \circ \cong} (A \ltens C) \lseq (B \ltens D) \]
which is therefore principally determined by $q$ via duality.
The importance of condition \ref{it:coherence} from the above definition is
that it may be understood as specifying a coherence property. Note that
coherence for compact closed categories is already known, so the faithfulness
of $U$ and the fact that $U$ preserves all the relevant structure provide us
with a simple criterion to check if a diagram in a strong \bvcat
commutes. Less formally, condition \ref{it:coherence} allows us to inherit many of the coherence
properties of the compact closed category $\DD$ (which are already known) into
$\CC$.

\begin{restatable}{theorem}{strongBVcategory}\label{thm:strong-bv-category}
  Let $\tuple{\CC,\vlte,\vlse,\vlun,\lneg\cdot,U}$ be a strong \bvcat. Then $\tuple{\CC,\vlte,\vlse,\vlun,\lneg\cdot,U}$ is also a \bvcat (with negation) in the sense of \cite{blu:pan:slav:deep} and \cite{bv-cats2}.
\end{restatable}
\begin{proof}
  (Sketch.)
  The proof follows easily from the fact that $U$ is faithful and strictly
  preserves all the relevant data which allows us to reduce the commutativity
  of all the relevant diagrams in $\CC$ to their commutativity in $\DD$ under
  $U$. The latter is very easy to verify using the string diagram language (or
  coherence theorem) for strict compact closed categories.

  To illustrate this, let us show why diagram  \eqref{eq:assoc-w-seq} from Appendix \ref{app:BVcat} commutes in $\CC.$
  Since $U$ is faithful and strictly preserves the two monoidal structures $\ltens$ and $\lseq$, it follows that \eqref{eq:assoc-w-seq}, i.e.,
  \[
    \begin{tikzcd}
      ((A\lseq B) \otimes (C \lseq D)) \otimes (E \lseq F) \arrow[r, "\alpha"] \arrow[d, "q' \otimes \mathrm{id}"']             & (A\lseq B) \otimes ((C \lseq D) \otimes (E \lseq F)) \arrow[d, "\mathrm{id} \otimes q'"]      \\
      ((A \otimes C) \lseq (B\otimes D)) \otimes (E \lseq F) \arrow[d, "q'"']                                 & (A\lseq B) \otimes ((C \otimes E) \lseq (D \otimes F)) \arrow[d, "q'"] \\
      ((A \otimes C) \otimes E) \lseq ((B \otimes D) \otimes F)  \arrow[r, "\alpha \lseq \alpha"'] & (A \otimes (C \otimes E)) \lseq (B \otimes (D \otimes F))
      \end{tikzcd}
    \]
    commutes in $\CC$, iff the diagram
  \[
    \begin{tikzcd}
      UA\boxtimes UB \boxtimes UC \boxtimes UD \boxtimes UE \boxtimes UF \arrow[r, "\id"] \arrow[d, "\id \boxtimes \sigma \boxtimes \id \boxtimes \mathrm{id}"']             & UA\boxtimes UB \boxtimes UC \boxtimes UD \boxtimes UE \boxtimes UF \arrow[d, "\mathrm{id} \boxtimes \id \boxtimes \sigma \boxtimes \id"]      \\
      UA \boxtimes UC \boxtimes UB\boxtimes UD \boxtimes UE \boxtimes UF \arrow[d, "\id \boxtimes \sigma \boxtimes \id"']                                 & UA\boxtimes UB \boxtimes UC \boxtimes UE \boxtimes UD \boxtimes UF \arrow[d, "\id \boxtimes \sigma \boxtimes \id"] \\
      UA \boxtimes UC \boxtimes UE \boxtimes UB \boxtimes UD \boxtimes UF  \arrow[r, "\id"'] & UA \boxtimes UC \boxtimes UE \boxtimes UB \boxtimes UD \boxtimes UF
      \end{tikzcd}
    \]
    commutes in $\DD$ (which we recall is \emph{strict} compact closed) and this follows immediately from the coherence theorem of symmetric monoidal categories.
    
    The remaining diagrams can be checked in a similar way.
\end{proof}

Note that the definition of strong \bvcat requires checking considerably fewer coherence
diagrams compared to that of \bvcat (see Appendix \ref{app:BVcat}). This is largely thanks
to condition \ref{it:coherence}, which is indeed a strong assumption.

\section{Denotational Semantics}\label{sec:semantics}

We now describe the denotational semantics of $\BV$ in an arbitrary, but fixed,
strong \bvcat with $U\colon\CC \hookrightarrow \DD$ the forgetful
inclusion. Because of Definition~\ref{def:bv-model}.\ref{it:units}, we may
choose $\lunit$ to be the monoidal unit for the $\lpar$ monoidal structure as
well and then we can define $\lneg{(\cdot)} = (\cdot)\limp \vlun$.
We do so in the sequel for simplicity. A $\BV$-formula $A$ is interpreted as an object $\sem{A} \in \mathrm{Ob}(\CC)$ as follows:
\begin{equation}\label{eq:semForms}
  \hfill
  \sem{\lunit} = \lunit
\qquad
  \sem{a} = c_a
\qquad
  \sem{\lneg a} = \lneg{\sem{a}}
\qquad
  \sem{A \odot B} = \sem A \odot \sem B
  \hfill
\end{equation}
where $c_a$ is some chosen object of $\CC$, and $\odot \in \set{ \lpar, \lseq, \ltens }$.
A $\BV$-formula $A$ also admits an interpretation $\afsem{A} \in \mathrm{Ob}(\DD)$ which can be defined by interpreting all tensors as $\boxtimes$:
\begin{equation}\label{eq:afsemForms}
  \hfill
  \afsem{\lunit} = J
\qquad
  \afsem{a} = Uc_a
\qquad
  \afsem{\lneg a} = \lneg{\afsem{a}}
\qquad
  \afsem{A \odot B} =\afsem A \boxtimes \afsem B
  \hfill
\end{equation}
where again $\odot \in \set{ \lpar, \lseq, \ltens }$.
It now immediately follows that $U \sem{A} = \afsem{A},$ because the functor $U$ strictly preserves all this data.
The interpretation of formula equivalence $A \feq B$ is given by a (symmetric) monoidal natural isomorphism $\sem{A \feq B} \colon \sem A \cong \sem B$ in $\CC$
defined in the obvious way. Likewise, formula equivalence can be interpreted in $\DD$ as a symmetric monoidal natural isomorphism $\afsem{A \feq B} \colon \afsem A \cong \afsem B$ in $\DD$. Note that
in $\DD$, all the isomorphisms are identities, except for the two symmetries of $\lpar$ and $\ltens.$

\begin{definition}
  \label{def:bv-semantics}
  Let $\dD=\nosmash{\odv{A}{\dD}{B}{\SBV}}$ be a derivation in $\SBV$ with premise $A$ and conclusion $B$.
  The interpretation of $\dD$ in $\CC$ is the morphism $\sem{\dD} \colon \sem{A} \to \sem{B}$ defined as follows:
  $$
  \sem{\Phi \odot \Psi} \coloneqq \sem{\Phi} \odot \sem{\Psi}
  \qquad 
  \sem{\vldownsmash{\odn{\Phi}{\rR}{\Psi}{}}} \coloneqq \sem{\Psi} \circ \sem{\rR} \circ \sem{\Phi}
  \qquad
  \sem{\nodt{\;\;A\;\;}{\feq}{B}{}} \coloneqq \sem A \xrightarrow{\sem{A \feq B}} \sem B
  $$
  where  $\sem{\rR}$ is defined as follows
      \begin{equation}
    \adjustbox{max width=.94
      \textwidth}{$\hskip-.8em
    \begin{array}{l}
      \sem{\nodn{\lunit}{\aidr}{a\lpar \lneg a}{}} 
    \;\coloneqq\;
      \lunit 
    \xrightarrow{\Lambda(\mathrm{ev} \circ \swap \circ l)} 
      \lneg{\Big(\lneg{\sem a} \otimes \sem a^{\perp\perp}\Big)}  
      \quad\qquad 
      \left(\text{where $\lneg{\Big(\lneg{\sem a} \otimes \sem a^{\perp\perp}\Big)} = \sem{a\lpar \lneg a} $ by definition}\right)
    \\[15pt]
      \sem{\nodn{a\ltens \lneg a}{\aiur}{\lunit}{}} 
    \;\coloneqq\;
      \Big(\sem a \otimes \lneg{\sem a} \Big)
    \xrightarrow{\swap} 
      \Big(\lneg{\sem a} \otimes \sem a \Big)
    \xrightarrow{\mathrm{ev}} 
      \lunit 
    \\[15pt]
      \sem{{\nodn{A\ltens (B\lpar C)}{\swir}{(A\ltens B)\lpar C}{}}} 
    \;\coloneqq\;
      \Big(\sem A\ltens (\sem B\lpar \sem C) \Big)
    \xrightarrow{\Lambda(\mathrm{ev} \circ (\id \otimes \mathrm{ev}) \circ \cong)}
      \Big(\lneg{(\sem A \ltens \sem B)} \multimap \sem C \Big)
    \xrightarrow{\cong} 
      \Big((\sem A\ltens \sem B)\lpar \sem C\Big)
    \\[15pt]
      \sem{\nodn{(A\lpar B)\lseq (C\lpar D)}{\qdr}{(A\lseq C)\lpar(B\lseq D)}{}} 
    \;\coloneqq\;
      \Big((\sem A\lpar \sem B)\lseq (\sem C\lpar \sem D) \Big)
    \xrightarrow{q} 
      \Big((\sem A\lseq \sem C)\lpar(\sem B\lseq \sem D)\Big)
    \\[15pt]
      \sem{\nodn{(A\lseq B)\ltens(C\lseq D)}{\qur}{(A\ltens C)\lseq(B\ltens D)}{}}
    \;\coloneqq\;
      \Big((\sem A\lseq \sem B)\ltens( \sem C\lseq \sem D) \Big)
    \xrightarrow{q'} 
      \Big((\sem A \ltens \sem C) \lseq ( \sem B \ltens \sem D) \Big) 
    \end{array}$}
  \end{equation}
  where $\odot \in \set{ \lpar, \lseq, \ltens }$;
  $\Lambda_{ABC} \colon \CC(A \otimes B, C) \cong \CC(A, B \multimap C)$ is the usual natural bijection;
  $\mathrm{ev}_{AB} \colon (A \multimap B) \otimes A \to B$ is the evaluation morphism in $\CC$;
  $l_A \colon \lunit \otimes A \cong A$ is the left unitor;
  the left unnamed isomorphism in the $\swir$ rule has type $(\lneg{\sem B} \multimap \sem C) \ltens (\sem A \multimap \lneg{\sem B}) \ltens \sem A \cong \sem A \ltens (\sem B \lpar \sem C) \ltens
  \lneg{(\sem A \ltens \sem B)}$ 
  and it is constructed using the obvious isomorphisms from the $*$-autonomous structure.

  The interpretation of $\dD$ in $\DD$ is the morphism $\afsem{\Phi} \colon \afsem A \to \afsem{B}$ defined as follows:
  \begin{equation}
    \begin{array}{l}
      \afsem{\Phi \odot \Psi} 
    \coloneqq
      \afsem{\Phi} \boxtimes \afsem{\Psi}
    \qquad
      \Afsem{\vldownsmash{\odn{\Phi}{\rR}{\Psi}{}}} 
    \coloneqq
      \afsem{\Psi} \circ \afsem{r} \circ \afsem{\Phi}
    \qquad
      \Afsem{\nodt{\;\;A\;\;}{\feq}{B}{}} 
    \coloneqq
      \afsem A \xrightarrow{\afsem{A \feq B}} \afsem B
    \\[15pt]
      \Afsem{\nodn{\lunit}{\aidr}{a\lpar \lneg a}{}}
    \coloneqq
      \ccunit \xrightarrow{\eta_{\afsem a^\perp}} \afsem{a} \boxtimes \lneg{\afsem a} 
    \qquad
      \Afsem{\nodn{a\ltens \lneg a}{\aiur}{\lunit}{}} 
    \coloneqq
      \afsem a \boxtimes \lneg{\afsem a} \xrightarrow{\epsilon_{\afsem a}} \ccunit
    \\[15pt]
      \Afsem{\nodn{A\ltens (B\lpar C)}{\swir}{(A\ltens B)\lpar C}{}} 
    \coloneqq
      \afsem A \boxtimes \afsem B \boxtimes \afsem C \xrightarrow{\id} \afsem A \boxtimes \afsem B \boxtimes \afsem C
    \\[15pt]
      \Afsem{\nodn{(A\lpar B)\lseq (C\lpar D)}{\qdr}{(A\lseq C)\lpar(B\lseq D)}{}} 
    \coloneqq
      \afsem A \boxtimes \afsem B \boxtimes \afsem C \boxtimes \afsem D 
    \xrightarrow{\id \boxtimes \swap \boxtimes \id} 
      \Afsem A \boxtimes \afsem C \boxtimes \afsem B \boxtimes \afsem D
    \\[15pt]
      \Afsem{\nodn{(A\lseq B)\ltens(C\lseq D)}{\qur}{(A\ltens C)\lseq(B\ltens D)}{}} 
    \coloneqq
      \afsem A \boxtimes \afsem B \boxtimes \afsem C \boxtimes \afsem D 
    \xrightarrow{\id \boxtimes \swap \boxtimes \id} 
      \afsem A \boxtimes \afsem C \boxtimes \afsem B \boxtimes \afsem D
    \\[15pt]
    \end{array}
  \end{equation}
  where $\odot \in \set{ \lpar, \lseq, \ltens }$, and where $\eta_A \colon \ccunit \to A^\perp \boxtimes A$ and $\epsilon_A \colon A \boxtimes A^\perp \to \ccunit$ are the unit and counit of the compact closed structure.
\end{definition}

\begin{remark}\label{rem:af-semantics}
  The definition of $\afsem{\dD}$  can be seen as directly interpreting the atomic flow of $\dD$ as a string diagram representing the morphism in a strict compact closed category $\DD$,
  where
  $\begin{AFz}{\naf{0,1}{\nafhdx{}{}{}{}{.5}}}\end{AFz}$ represents the morphism $\eta$ and
  $\begin{AFz}{\naf{0,0}{\nafhux{}{}{}{}{.5}}}\end{AFz}$ represents the morphism $\epsilon$, which are usually drawn as cups $\begin{AFz}{\naf{0,0}{\nafcapshort1}}\end{AFz}$  and caps $\begin{AFz}{\naf{0,1}{\nafcupshort1}}\end{AFz}$, respectively, in string diagrams (see, e.g., \cite{string:diagrams:RTII}).
\end{remark}

The relationship between the two semantic interpretations is given by the following lemma.
\begin{lemma}\label{lem:interpretations}
  For any 
  $\SBV$-derivation
  $\dD=\vldownsmash{\odv{A}{\dD}{B}{\SBV}}$, the following diagram commutes in $\DD$.
  \begin{equation}\label{eq:semantic-lemma}
    \hfill
    \begin{tikzcd}
      U\sem A \arrow[d, "U\sem{\Phi}"'] \arrow[equal, r] 
    & 
      \afsem{A} \arrow[d, "\afsem{\Phi}"]
    \\
      U\sem{B} \arrow[equal, r] 
    & 
      \afsem{B}
    \end{tikzcd}
    \hfill
  \end{equation}
\end{lemma}
\begin{proof}
  This follows by induction on $\Phi$. The two most interesting cases are $\qdr$ and $\qur$. The case for $\qdr$ follows immediately from the definition of a strong \bvcat.
  The case for $\qur$ follows easily, because $U \sem{\qur} = U(\cong) \circ U(\sem{q}^\perp) \circ U(\cong) = \id \circ (\id \boxtimes \swap^\perp \boxtimes \id) \circ \id = \id \boxtimes \swap \boxtimes \id
  = \afsem{\qur}$
  where we used the fact that $U$ strictly preserves all the relevant structure and maps the two unnamed isomorphisms to identities in $\DD$.
  The remaining cases are straightforward and follow easily from standard results about $*$-autonomous and compact closed categories. More details can be found in Appendix~\ref{app:lem21}.
\end{proof}

\begin{proposition}
  \label{prop:yanking-semantics}
  $\Phi$ and $\Phi'$ be two $\SBV$-derivations
  such that $\AFof{\Phi'} = \Yof{\AFof{\Phi}}.$
  Then $\afsem{\Phi} = \afsem{\Phi'}.$
\end{proposition}
\begin{proof}
  This follows easily from \cref{rem:af-semantics}, because the semantic interpretation in the strict compact closed category $\DD$ is sound (i.e., invariant) with respect to yanking.
\end{proof}

\begin{theorem}[Soundness]\label{thm:soundness}
  Let $\Phi$ and $\Phi'$ be derivations with $\Phi \CEto \Phi'$. Then $\sem{\Phi} = \sem{\Phi'}$.
\end{theorem}
\begin{proof}
  Since $U$ is a faithful functor, it suffices to prove that $U\sem{\Phi} = U\sem{\Phi'}.$
  By Lemma \ref{lem:interpretations}, it suffices to prove that $\afsem{\Phi} = \afsem{\Phi'}$.
  Using Theorem \ref{thm:YAF}, we know that $\AFof{\dDp}=\Yof{\AFof{\dD}}$ and Proposition \ref{prop:yanking-semantics}
  concludes the proof.
\end{proof}

\begin{corollary}\label{cor:proof-equivalence}
  If $\Phi \preq \Psi$, then $\sem{\Phi} = \sem{\Psi},$ i.e., the semantic interpretation is invariant with respect to proof equivalence.
\end{corollary}
\begin{proof}
  By the previous theorem, it suffices to prove this for normalized proofs. Since $U$ is faithful, it suffices to prove
  $U\sem{\Phi} = \afsem{\Phi} = \afsem{\Psi} = U \sem{\Psi}.$ This now follows trivially from Proposition \ref{prop:yanking-semantics}.
\end{proof}

\section{Concrete Models of BV}\label{sec:concrete}

In this section we describe several categories that have the structure of a
strong \bvcat, thus providing concrete models of the logic $\BV$. All the
concrete models that we consider are based on existing categories that have
already been studied in prior work. Because of the strictness requirement of
Definition \ref{def:bv-model}, a common pattern in three of the models that we
consider is that we take the \emph{skeleton} subcategory of a category that has
already been considered in the literature. This allows us to easily satisfy the
strictness requirement. However, we conjecture
that the strictness requirement is not necessary -- see Section
\ref{sec:conclusion} for more discussion related to this.

The most obvious, but rather degenerate, class of models is given by strict
compact closed categories (Subsection \ref{sub:strict}) where one can interpret
all three multiplicative connectives of $\BV$ using the same monoidal tensor. The
most mathematically interesting and natural example is based on
finite-dimensional operator spaces (Subsection \ref{subsec:OOS}), with strong
links to quantum theory, where mathematicians have identified three different
operator space tensor products that allow us to interpret $\BV$. Our two remaining
examples are based on gluing and orthogonality \cite{HS}: the Caus[-]
construction on a category of completely positive maps (Subsection
\ref{sub:caus}) and finite-dimensional probabilistic coherence spaces
(Subsection \ref{sub:pcs}). The former has strong links to quantum
theory whereas the latter is relevant to classical probabilistic computation.

\subsection{Strict Compact Closed Categories}
\label{sub:strict}
Every strict compact closed category can be seen as a strong \bvcat by taking $\CC = \DD$ and $U$ to be the identity functor.

A relevant example is the \emph{skeleton subcategory of finite-dimensional vector spaces} $\FdVectsk$ whose 
objects are vector spaces of the form $\mathbb C^n$ and 
morphisms $\FdVectsk(\mathbb C^n , \mathbb C^m)$ are $m\times n$ complex matrices, 
where the multiplication of matrices serves as composition of morphisms, 
the tensor product $\boxtimes$ is the Kronecker product of matrices (with unit $\mathbb C$), and 
where duality is defined as transposition of matrices, denoted $\transpose$.
Note that $\FdVectsk \simeq \FdVect$, because every $m \times n$ matrix
determines a linear map $\mathbb C^n \to \mathbb C^m$ (with respect to the
standard basis), so we may view the morphisms as linear maps as well.

\begin{theorem}
  The tuple $\tuple{\FdVectsk,\boxtimes,\boxtimes,\mathbb C,\transpose,\catid}$ is a strong \bvcat.
\end{theorem}
\begin{proof}
  It is well-known that $\FdVectsk$ is \emph{strict} compact closed. See for example \cite{cqm}. 
\end{proof}

\subsection{Finite-dimensional Operator Spaces}\label{subsec:OOS}
\def\nmatof#1{\mathbf M_{#1}}
\def\smatof#1#2{\mathbb M_{#2}(#1)}
\def\normm#1#2#3{\norm{#1}_{\smatof{#2}{#3}}}
\def\oss{\mathcal O}

A \emph{finite-dimensional operator space} \cite{er2000operator} is a pair $\left(X, \oss \right)$ consisting of
a finite-dimensional vector space $X$ and
an \emph{operator space structure} $\oss= \Set{\normm{\cdot}Xn \mid n \in \mathbb N }$ given by a sequence of norms $\normm{\cdot}Xn \colon \smatof Xn \to \mathbb R_+$ on the vector space of $n\times n$ matrices with entries in $X$ that satisfy specific axioms (omitted here). 
If $(X, \oss)$ and $(Y, \oss')$ are two operator spaces, 
we say that a linear map $f \colon X \to Y$ is a \emph{complete contraction} \cite{er2000operator}
if $\normm{f (A)}Yn \leq \normm{A}Xn$ for every $n \in \mathbb N$ and $A \in \smatof Xn$, where by $f(A)$ we denote the matrix obtained by applying $f$ to each entry of $A$.

The category $\FdOS$ of finite-dimensional operator spaces and (linear) complete contractions was recently studied in \cite{qcs}, where it was shown that it is a \bvcat.
Duals (in the categorical/logical sense) coincide with the operator space duals \cite{er2000operator}, written~$\cdot^*$ in $\FdOS$. The \emph{projective tensor} $\ptimes$ \cite[\S 7.1]{er2000operator} serves as the multiplicative conjunction $\ltens$,
the \emph{injective tensor} $\itimes$ {\cite[(1.5.1)]{blecher-merdy}} can be identified with the multiplicative disjunction $\lpar,$
and the \emph{Haagerup tensor} $\hagertimes$ \cite[\S 9]{er2000operator}, \cite[pp. 30--34]{blecher-merdy}, \cite[\S 5]{Pisier_2003} as the seq~$\lseq$ of~$\BV$. By \cite[Theorem~6.1]{er-shuffle} (see also \cite{qcs}), we can define a complete contraction 
  \[
  \hfill
  q_{ABCD} \colon (A \itimes B) \hagertimes (C \itimes D) 
  \to
  (A \hagertimes C) \itimes (B \hagertimes D)
  \hfill
  \]
which gives us the $q$ natural transformation. It is well-known that the Haagerup tensor is not symmetric and it is self-dual in $\FdOS$ \cite[\S 9]{er2000operator},
thus giving us $\kappa_{AB} \colon A^* \hagertimes B^* \cong (A \hagertimes B)^*$.

The \defn{skeleton subcategory} $\FdOSsk$ is the subcategory of $\FdOS$ whose objects are operator spaces of the form $(\mathbb C^n, \oss)$ for some operator space structure $\oss$ on $\mathbb C^n$
and whose morphisms are the complex matrices that represent linear complete contractions between such spaces. All of the aforementioned constructions can be easily adapted to $\FdOSsk$ and the obvious
forgetful functor $U \colon \FdOSsk \to \FdVectsk$ strictly preserves all this data.

\begin{theorem}
  The tuple $\tuple{\FdOSsk,\ptimes,\hagertimes, \mathbb C, \cdot^*, U}$ is a strong \bvcat.
\end{theorem}
\begin{proof}
  Straightforward verification by using (the adapted) results from \cite{qcs}.
\end{proof}

The theory of operator spaces has strong links to quantum theory. Von Neumann 
algebras (e.g., $B(H)$ -- the algebra of bounded operators on a Hilbert space 
$H$) are operator spaces and so are their preduals (e.g., $T(H)$ -- the trace 
class operators on a Hilbert space $H$). The former can be used to formulate 
quantum computation in the Heisenberg picture of quantum theory, whereas the 
latter can be used to formulate quantum computation in the Schrödinger picture. 
The authors in \cite{qcs} use the category $\FdOS$ to construct a model of MALL 
(Multiplicative Additive Linear Logic) in which the polarised linear logic 
duality coincides with the Heisenberg-Schrödinger duality of quantum theory. The 
Haagerup tensor, which is used to interpret the seq connective $\lseq$ of $\BV$, 
plays an essential role in this development. The work in \cite{qcs} builds on
prior work \cite{os-lics} on (infinite-dimensional) operator spaces, linear logic
and the Heisenberg-Schrödinger duality where the authors also consider the Haagerup tensor
and discuss its relation to the logic $\BV$ and some of its other properties relevant to quantum computation. 

\subsection{Completely Positive Maps and the Caus[-] Construction}
\label{sub:caus}
\newcommand{\CPs}{\ensuremath{\mathbf{CP}^*}}

In \cite{sim:kiss:BV}, the authors describe a categorical construction, called
$\Caus[-]$, that produces a new $*$-autonomous category $\Caus[\CC]$ with some
additional structure (in particular an extra tensor $<$) from a given compact closed category $\CC$
satisfying specific conditions. This allows them to construct a \bvcat
where the tensor $<$ represents the connective $\lseq$ of $\BV$.
There is an evident forgetful functor $U \colon \Caus[\CC] \to \CC$ which strictly preserves
all of the relevant structure.
The main concrete example is the category
$\Caus[\CPs[\mathbf{FHilb}]]$, where $\mathbf{FHilb}$ is the category
of finite-dimensional Hilbert spaces with linear maps as morphisms and $\CPs$ is another categorical
construction which gives us a category of completely-positive maps in this case.

If we consider the skeletal subcategory $\mathbf{FHilb_{sk}}$ whose objects are the Hilbert spaces $\mathbb C^n$ and the morphisms are the complex matrices,
then we can construct a strong \bvcat.

\begin{restatable}{theorem}{thmCaus}\label{thm:caus}
  The tuple $\tuple{\Caus[\CPs[\mathbf{FHilb_{sk}}],\otimes, <, \lunit, \cdot^\ast, U}$
  is a strong \bvcat.
\end{restatable}
\begin{proof}
  Straightforward verification using the results from \cite{sim:kiss:BV} and \cite{cqm} (see Appendix~\ref{app:thmCaus} for more details).
\end{proof}

This example also has strong links to quantum theory. The category
$\CPs[\mathbf{FHilb}]$ is equivalent to the category of finite-dimensional von Neumann algebras (equivalently finite-dimensional C*-algebras) and completely positive maps
between them. The category $\Caus[\CPs[\mathbf{FHilb}]]$ then gives a model of higher-order quantum computation (for finite-dimensional systems).

\subsection{Probabilistic Coherence Spaces}
\label{sub:pcs}
\def\pcsbracket#1#2{\langle #1 , #2 \rangle}
\def\webof#1{\left|#1\right|}
\def\Pcoh{\mathbf{Pcoh}}
\def\skPcoh{\mathbf{P_{sk}}}
\def\funof#1{P_#1}
\newcommand\FdPcoh{\mathbf{P_{sk}}}

Our next example is based on a category of finite-dimensional probabilistic coherence spaces \cite{gir:quantum}.
The category is described in \cite{blu:pan:slav:deep} and it is $*$-autonomous. The authors in \cite{blu:pan:slav:deep} also define a third tensor product that we use.
We adapt the results from \cite{blu:pan:slav:deep} by considering the skeletal versions
of these constructions.
Let $\FdPcoh$ be the category whose objects are pairs $A=(\mathbb R_+^n,\funof A)$, where $n \in \mathbb N$ and $\funof A\subseteq \mathbb R^{n}_+$,
where $\mathbb R_+$ is the set of non-negative reals, such that:
\begin{itemize}
  \item $\funof A=\funof A^{\lbot\lbot}$, where $\funof A^\lbot=\set{w \in \mathbb R_+^{n} \mid \pcsbracket wv\coloneqq\sum_{i=1}^n w_i v_i\leq 1 \text{ for all } v \in \funof A}$;
  \item and certain other conditions hold, but we omit them here for brevity.
\end{itemize}

The homsets $\skPcoh\big((\mathbb R_+^n,\funof A), (\mathbb R_+^m,\funof B)\big)$ are given by the
$m\times n$ matrices with coefficients in $\mathbb R_+$, such that for all $v\in \funof A$
we have that $Mv\in \funof B$, where $Mv$ stands for the multiplication of the matrix $M$ with the vector $v.$
We write $\vlte$ for the tensor product in $\FdPcoh$ which can be adapted in an obvious way from the usual definition in \cite{dan:ehr:prob,blu:pan:slav:deep,gir:quantum}.
We write $\oslash$ for the (again obvious) adapation of the tensor defined in \cite{blu:pan:slav:deep}
to $\FdPcoh$. We note that the action of $\otimes$ and $\oslash$ on morphisms coincides with the Kronecker product of matrices.
Duals are given by $(\mathbb R_+^n, \funof A)^\perp = (\mathbb R_+^n, \funof A^\perp)$ and their action on morphisms is given by transposition of matrices.
Let $\mathbf{Mat}_{\mathbb R_+}$ be the category whose objects are given by $\mathbb R_+^n$, for $n \in \mathbb N$, and whose homsets
$\mathbf{Mat}_{\mathbb R_+}(\mathbb R^n_+, \mathbb R^m_+)$ are given by the $m\times n$ matrices with coefficients in $\mathbb R_+.$ This category
is strict compact closed with all the constructions fully analogous to $\FdVectsk.$ We
write $U \colon \FdPcoh \hookrightarrow \mathbf{Mat}_{\mathbb R^+}$ for the obvious forgetful functor.

\begin{theorem}
  The tuple $\tuple{\skPcoh,\vlte,\oslash,\mathbb R^+,\lneg\cdot,U}$
  is a strong \bvcat.
\end{theorem}
\begin{proof}
  Straightforward verification using the (adapted) results from \cite{blu:pan:slav:deep}.
\end{proof}

Probabilistic coherence spaces have been studied in the context of (classical) probabilistic computation \cite{pcbv-full-abstraction,ppcf-full-abstraction}.
Note that the examples in the previous two subsections also support probabilistic effects.

\section{Conclusion and Future Work}\label{sec:conclusion}

In this paper we proposed a notion of \emph{proof identity} for $\BV$, and for this we introduced \emph{atomic flows} for $\BV$. We have shown that the global operation of cut elimination in $\BV$ via \emph{splitting} coincides on atomic flows with the local operation of \emph{yanking}. This allowed us to provide a refined notion of \emph{strong \bvcat} that is sound for the logic $\BV$, and we provided some concrete non-trivial examples of such categories.

This work brings together two communities: those who study the proof theory of $\BV$ and those who study its categorical semantics. This opens up new directions of research on which we elaborate below.

\subparagraph*{Relation to Pomset Logic.}
An independent justification for our choice of proof identity for $\BV$ is that it coincides with the one induced by the proof nets of pomset logic. It has been shown in~\cite{tito:lutz:csl22,tito:str:SIS-III}, that every $\BV$-proof can be translated into a correct pomset logic proof net. In the cut-free case, this is just the formula in the conclusion together with the axiom links. But the atomic flow of a cut-free $\BV$-derivation with premise $\vlun$ is also just the axiom links on the conclusion. Therefore, two $\BV$ proofs have the same atomic flow if and only if they have the same pomset logic proof net. This raises some questions for future research: What is the precise relation between cut elimination in pomset logic proof nets and yanking in atomic flows? And can we have a correctness criterion for $\BV$ atomic flows? As shown in~\cite{tito:lutz:csl22,tito:str:SIS-III}, the correctness criterion for pomset logic is not suitable for $\BV$.

\subparagraph*{Strong \bvcats vs \bvcats.}
We proved that every strong \bvcat is a \bvcat in the
sense of~\cite{blu:pan:slav:deep,bv-cats2}. We conjecture that the 
converse is not true, i.e. there exists a \bvcat which cannot be 
embedded into a compact closed one in the way that we require. The 
embedding of a strong \bvcat into a strict compact closed category ensures that 
we get all the coherence properties for the proof of soundness 
(Theorem \ref{thm:soundness}) and invariance with respect to proof 
equivalence (Corollary \ref{cor:proof-equivalence}), for which we 
also use results based on atomic flows.  
Even though the original definition of a \bvcat gives us many coherence diagrams (see Appendix \ref{app:BVcat}), we do not know if it gives us \emph{all} the ones we need, because the splitting lemma (used for cut elimination / proof normalisation) is difficult to work with categorically, and we doubt that \bvcats give us enough coherence diagrams for atomic flows in relation to Corollary \ref{cor:proof-equivalence}. 

We would also like to formulate a more general notion of proof equivalence 
compared to the one that we used here (which is based on atomic flows) and then check if the original definition of \bvcat gives us enough coherence diagrams for this notion of proof equivalence.
If not, then it would be important to identify the right set of coherence laws for \bvcats and modify the definition of a \bvcat accordingly. Another open problem related to this is whether
one can formulate a coherence theorem for \bvcats (in a similar way to Mac Lane's coherence theorem for monoidal categories \cite{mac-categories}).
The massive complexity discrepancy between pomset logic and $\BV$ suggests that this is a non-trivial question.

\subparagraph*{Yanking without Cut Elimination.} 
Our main result (\Cref{thm:YAF}) only speaks about derivations with premise $\vlun$. However, we conjecture that this can be generalized as follows:
\begin{conjecture}\label{con:YAF}
  Given a derivation $\vldownsmash{\odv A\dD B\SBV}$, then there is a derivation  $\vldownsmash{\odv A{\dDp} B\SBV}$, such that $\AFof{\dDp}=\Yof{\AFof{\dD}}$.
\end{conjecture}
This would then also strengthen our second result~\Cref{thm:soundness}. We also conjecture that the methods that we need to develop to prove~\Cref{thm:YAF} would also help to make progress towards a coherence theorem for \bvcats as discussed above.

\subparagraph*{Strictness.} 
\Cref{def:bv-model} has a strictness condition on the compact closed category
$\catD$ and on the forgetful functor $U\colon\catC\to\catD$. One reason for
this is that in atomic flows the associativity and unitality of the tensor is ``on the nose'' and it is not
immediately clear how to relax the strictness condition. It would certainly
require checking many coherence conditions and possibly modifying the proofs in \Cref{sec:splitting}, but we conjecture that \Cref{thm:soundness} can also be extended to the
non-strict case. In particular, we conjecture that the non-strict versions of the concrete categorical models in \Cref{sec:concrete} are sound.
Another reason for the strictness assumption is that it is currently unclear (to us) what is a good notion of a (non-strict) $*$-autonomous functor (compare \cite{star-aut-functor} and \cite{bv-cats2}).
We conjecture that the strictness assumption on the functor $U$ can be removed, but this requires performing a more careful analysis of all of the coherence properties that are required to hold in relation to it. 

\subparagraph*{Models of Intuitionistic $\BV$.}
Another direction for future work is to investigate suitable categorical models
for Intuitionistic $\BV$ \cite{acc:str:IBV,acc:str:IBVext}. It would be
interesting to see if natural models can be found by considering categories of
(possibly infinite-dimensional) operator spaces \cite{os-lics,category-os}
where the seq connective $\lseq$ of $\BV$ is interpreted by variants of the
Haagerup tensor (see \cite{er-shuffle}).

\bibliography{biblio}

\clearpage
\appendix

\section{Omitted Proofs for Section~\ref{sec:splitting}}\label{app:1}

\lemSplit*
\begin{proof}
  We recall that the standard proof of the splitting lemma proceeds by a case analysis on the last rule of the derivation $\dD$, and applying the induction hypothesis on the derivation with conclusion the premise of the last rule.
  The additional observation here is that each transformation preserves the atomic flow.
  This follows from the fact that the only rule in $\BV$ whose corresponding atomic flow is not an identity flow is the $\aidr$ rule, and the only case in which it is involved is in the third item above, where the atomic flow is the same modulo the identity over flows given by the fact that string diagrams are considered up-to the equivalences in the right of \Cref{fig:af}, in particular, the ones allowing us to freely `slide' boxes over and under wires. 

  \begin{enumerate}
    \item If there is a derivation $\odr{\dD}{(A\vlte B)\vlpa K}{\BV}$, then we have three cases:
    \begin{itemize}
      \item the bottom-most rule is applied to $K$, that is $\dD=\odr{\dD'}{(A\vlte B)\vlpa \left(\odn{K'}{\rrule}{K}{}\right)}{}$, in which case we can conclude by applying the induction hypothesis to $\dD'$, and then applying the same rule instance $\rrule$ to the resulting derivation, obtaining derivations 
      $$
      \dD[K]=\odP{\odi{\odd{\odh{K_A\vlpa K_B}}{\dD[K']}{K'}{}}{\rrule}{K}{}}
      \quand
      \odr{\dD[A]}{A\vlpa K_A}{}
      \quand
      \odr{\dD[B]}{B\vlpa K_B}{}
      $$
      We conclude since the atomic flow of $\swir$ is an identity flow, therefore we have
      \begin{equation}\label{eq:split:tens:deep}
        \hskip-1em
        \AFof{
          \odN{
            \odrP{\dD[A]}{A\vlpa K_A}{} \vlte \odrP{\dD[B]}{B\vlpa K_B}{}
            }{\swir}{
              (A\vlte B)
            \vlpa
              \odP{\odi{\odd{\odh{K_A\vlpa K_B}}{\dD[K']}{K'}{}}{\rrule}{K}{}}
            }{}
        }
      =
        \begin{AF}
          \naf{0,2}{\nafdiagramd{\AFof{\dD[A]}}{}{}21{}{}}
          \naf{5,2}{\nafdiagram{\AFof{\dD[B]}}{}{}21}
          \naf{5,-2}{\nafdiagramd{\AFof{\dD[K]'}}{}{,.,}21{}{}}
          \naf{5,-4.5}{\nafdiagramd{\AFof{\rrule}}{}{,.,}21{}{}}
          \naf{6,0}{\nafV{1}}
          \naf{-1,-2.5}{\nafV{3.5}}
          \naf{2.5,0}{\nafJr{1.5}1}
          \naf{1,-3.5}{\nafV{2.5}}
          \naf{2,0}{\nafJl{1.5}1}
        \end{AF}
      =
        \begin{AF}
          \naf{0,0}{\nafdiagram{\AFof{\dDp}}{}{}{2}1}
          \naf{1,-2.5}{\nafdiagram{\AFof{\rrule}}{,.,}{,.,}{1}1}
          \naf{-1,-2.5}{\nafV{1.5}}
          \naf{-1.5,-2.5}{\nafV{1.5}}
        \end{AF}
      =
        \AFof{\dD}
      \end{equation}

      \item the bottom-most rule is applied to $A$ or $B$, in which case we can conclude similarly to the previous case by applying the induction hypothesis to the derivation with conclusion the premise of the last rule, and then applying the same rule instance to the resulting derivation.

      \item the bottom-most rule is a $\feq$ applied to $A\ltens B$.
      In this case, we have that 
      $A=C_{i_1}\vlte \cdots\vlte C_{i_k}$ and $B=C_{j_1}\vlte \cdots\vlte C_{j_m}$ for some $i_1, \ldots, i_k, j_1, \ldots, j_m$ such that $\set{i_1, \ldots, i_k,j_1, \ldots, j_m}=\intset1n$, 
      and that
      $$
      \dD=\odr{\dD'}{\left(
          \od{\odo{\odh{
            \odP{\odo{\odh{C_1'}}{\feq}{C_1}{}}\vlte \cdots\vlte \odP{\odo{\odh{C_n'}}{\feq}{C_n}{}}
          }}{\feq}{
            \odP{\odo{\odh{C_{i_1}\vlte \cdots\vlte C_{i_k}}}{=}{A}{}}
            \vlte 
            \odP{\odo{\odh{C_{j_1}\vlte \cdots\vlte C_{j_m}}}{=}{B}{}}
          }{}}
        \right)\vlpa K}{}
      $$
      for some $C_k'\feq C_k$ for all $k\in\intset1n$.
      In this case, we repetively apply splitting to the conclusion of $\dD'$, until we get derivations
      $$
      \odv{K_1\vlpa \cdots \vlpa K_n}{\dD[K]}{K}{}
      \qquand
      \odr{\dD[C_k]}{C'_k\vlpa K_k}{}
      \text{ for all }k\in\intset1n
      $$
      And conclude since we can reconstruct 
      $$
      \dD[A]=\od{\odI{\odh{
        \odrP{\dD[C_{i_1}]}{C'_{i_1}\vlpa K_{i_1}}{}
        \vlte \cdots \vlte
        \odr{\dD[C_{i_k}]}{C'_{i_k}\vlpa K_{i_k}}{}
      }}{\swir}{
        \odP{\odo{\odh{C'_{i_1}\vlte \cdots\vlte C'_{i_k}}}{\feq}{A}{}}
      \vlpa
        K_{i_1}\vlpa \cdots \vlpa K_{i_k}
      }{}
      }
      $$
      and
      $$
      \dD[B]=\od{\odI{\odh{
        \odrP{\dD[C_{j_1}]}{C'_{j_1}\vlpa K_{j_1}}{}
        \vlte \cdots \vlte
        \odrP{\dD[C_{j_m}]}{C'_{j_m}\vlpa K_{j_m}}{}
      }}{\swir}{
        \odP{\odo{\odh{C'_{j_1}\vlte \cdots\vlte C'_{j_m}}}{\feq}{B}{}}
      \vlpa
        K_{j_1}\vlpa \cdots \vlpa K_{j_m}}{}
      }
      $$
      because we have that
      $$
      \AFof{
        \odN{
          \odrP{\dD[A]}{A\vlpa K_A}{} \vlte \odrP{\dD[B]}{B\vlpa K_B}{}
          }{\swir}{
            (A\vlte B)\vlpa\odvP{K_A\vlpa K_B}{\dD[K]}{K}{}
          }{}
      }
      =\qquad
      \begin{AF}
        \naf{4,2}{\nafdiagram{\AFof{\dD[C_1]\vlte \cdots\vlte \dD[C_n]}}{}{}{4.5}1}
        \naf{6.5,-2}{\nafdiagramd{\AFof{\dD[K]}}{}{}21{}{}}
        \naf{1.5,-2}{\nafdiagramd{\AFof{\feq}}{}{}21{}{}}
        \naflabel{5,.75}{\cdots}
        \naf{8,0}{\nafVu{1}{}{K_{C_n}}}
        \naf{0,0}{\nafVu{1}{C'_1}{}}
        \naf{5,0}{\nafJr{2}1}
        \naf{3,0}{\nafJl{2}1}
        \naf{0,-4}{\nafVu{1}{C_1}{}}
        \naf{3,-4}{\nafVu{1}{}{C_n}}
        \naf{7,-4}{\nafVu{1}{}{K}}
        \naflabel{1.2,0}{K_{C_1}}
        \naflabel{7,0}{C'_n}
        \naflabel{6.5,-.5}{\ldots}
        \naflabel{1.5,-.5}{\ldots}
        \naflabel{1.5,-4}{\ldots}
      \end{AF}
      \qquad=
      \AFof{\dD}
      $$
      where $K_A=K_{i_1}\vlpa\cdots\vlpa K_{i_k}$ and $K_B=K_{j_1}\vlpa\cdots\vlpa K_{j_m}$;

      \item or the bottom-most rule in $\dD$ is a $\swir$, that is 
      $\dD=\od{\odi{
        \odp{\dD'}{A\vlte (B\vlpa K)}{}
      }{\swir}{(A\vlte B) \vlpa K}{}}$, 
      in which case we can apply inductive hypothesis to $\dD'$, obtaining derivations
      $$
      \odv{K_B}{\dD[K]}{K}{}
      \qquand
      \odr{\dD[A]}{A}{}
      \qquand
      \odr{\dD[B]}{B\vlpa K_B}{}
      $$
      and conclude immediately since
      $$
      \AFof{
        \odn{
          \odrP{\dD[A]}{A}{} \vlte \odrP{\dD[B]}{B\vlpa K_B}{}
          }{\swir}{
            (A\vlte B)\vlpa\odvP{K_B}{\dD[K]}{K}{}
          }{}
      }
      =
      \AFof{
        \odn{
          \odrP{\dD[A]}{A}{} \vlte \odrP{\dD[B]}{B\vlpa \odvP{K_B}{\dD[K]}{K}{}}{}
          }{\swir}{
            (A\vlte B)\vlpa K
          }{}
      }
      =
      \AFof{\dD}
      $$
      because the atomic flow of $\swir$ is an identity flow.
    \end{itemize}
    
    \item If there is a derivation $\odr{\dD}{(A\vlse B)\vlpa K}{\BV}$, then we conclude similarly to the previous case.
    The only difference is that in this case we have to consider the rule $\qdr$ instead of the rule $\swir$, whose atomic flow is not an identity flow, but it is constructed by identities and crossings.
    
    \item If there is a derivation $\odr{\dD}{a\vlpa K}{\BV}$, then we have two cases:
    \begin{itemize}

      \item either the last rule is applied to $K$, that is $\dD=\odr{\dD'}{a\vlpa \left(\odn{K'}{\rrule}{K}{}\right)}{}$, in which case we can conclude by applying the induction hypothesis to $\dD'$, and then applying the same rule instance $\rrule$ to the resulting derivation, obtaining a derivation $\dD[a]=\odP{\odi{\odd{\odh{\lneg a}}{\dD[a]'}{K'}{}}{\rrule}{K}{}}$.
      We conclude since 
      $$
      \AFof{
        \odn{\vlun}{\aidr}{a\vlpa
        \odP{\odi{\odd{\odh{\lneg a}}{\dD[a]'}{K'}{}}{\rrule}{K}{}}
        }{}
      }
      =
      \AFof{
        \od{\odo{\odi{\odh{\vlun}}{\aidr}{a\vlpa
        \odP{\odd{\odh{\lneg a}}{\dD[a]'}{K'}{}}
        }{}}{}{a\vlpa \odnP{K'}{\rrule}{K}{}}{}
        }
      }
      =
      \begin{AF}
        \naf{0,0}{\nafdiagram{\AFof{\dDp}}{}{,,.,}{1.5}1}
        \naf{.5,-2.5}{\nafdiagram{\AFof{\rrule}}{}{,.,}{1}1}
        \naf{-1.5,-2.5}{\nafv{1.5}}
      \end{AF}
      =
      \AFof{\dD}
      $$

      \item or the last rule is applied to $a$, in which case it must be a $\aidr$ and 
      $\dD=
      \od{\odo{\odp{\dD'}{K'}{}}{}{\left(\odn{\vlun}{\rrule}{a\lpar \lneg a}{}\right) \vlpa K'}{}}
      $
      In this case, we conclude by letting $\dD[a]=\od{\odo{\odh{\lneg a \vlpa \odrP{\dD'}{K'}{}}}{}{K}{}}$ 
      since 
      $$
      \AFof{
        \odn{\vlun}{\aidr}{a\vlpa\odvP{\lneg a}{\dD[a]}{K}{}}{}
      }
      =
      \begin{AF}
        \naf{0,.75}{\nafhdx{a}{}{}{\lneg a}{1}}
        \naf{4,0}{\nafdiagram{\AFof{\dDp}}{}{,.,}{1.5}1}
      \end{AF}
      =
      \AFof{\dD}
      $$
      
    \end{itemize}
    
  \end{enumerate}
\end{proof}

\lemQuElim*
\begin{proof}
  We start by applying context reduction to $\dDp$, obtaining a formula $K$ and derivations
  \begin{equation}\label{eq:qUelimAF0}
    \hfill
    \begin{array}{c}
      \odr{\dDp[1]}{\left((A\vlte C)\vlse(B\vlte D)\right)\vlpa K}{\BV}
      \quand
      \odv{X\vlpa K}{\dDp[{\cC[X]}]}{\cC[X]}{\BV}
      \text{ for every formula } X
    \\[20pt]
      \text{such that}
      \qquad\qquad
      \begin{AF}
        \naf{0,0}{\nafdiagram{\AFof{\dDp}}{}{,.,}{1.5}1}
      \end{AF}
      =
      \begin{AF}
        \naf{0,2.5}{\nafdiagram{\AFof{\dDp[1]}}{}{,.,}41}
        \naf{0,0}{\nafdiagram{\AFof{\dDp[{\cC[(A\vlte C)\vlse(B\vlte D)]}]}}{,.,}{,.,}41}
      \end{AF}
    \end{array}
    \hfill
  \end{equation}
  We then apply the splitting lemma to $\oDp[1]$, obtaining formulas $K_l$ and $K_r$ and derivations 
  \begin{equation}\label{eq:qUelimAF1}
    \hfill
    \begin{array}{c}
      \odv{K_l\vlpa K_r}{\dD[K]}{K}{\BV}
      \text{ and }
      \odr{\dD[A\vlse B]}{(A\vlse B )\vlpa K_l}{\BV}
      \text{ and }
      \odr{\dD[C\vlse D]}{(C\vlse D) \vlpa K_r}{\BV}
      \quad\text{such that}
    \\[10pt]
      \qquad
      \AFof{\oDp[1]}
      =
      \AFof{
        \odN{
          \oD[A\vlse B]\vlte \oD[C\vlse D]
          }{\swir}{
            ((A\vlse B)\vlte(C\vlse D))\vlpa \oD[K]
          }{}
      }
      =
      \begin{AF}
        \naf{0,2}{\nafdiagramd{\AFof{\dD[A\vlse B]}}{}{}21{}{}}
        \naf{5,2}{\nafdiagram{\AFof{\dD[C\vlse D]}}{}{}21}
        \naf{5,-2}{\nafdiagramd{\AFof{\dD[K]}}{}{,.,}21{}{}}
        \naf{6,0}{\nafV1}
        \naf{-1,-1}{\nafV{2}}
        \naf{2.5,0}{\nafEXl{1.5}1}
        \naf{1,-2}{\nafV{1}}
      \end{AF}
    \end{array}
    \hfill
  \end{equation}
  Then, we can apply splitting on both $\oD[A\vlse B]$ and $\oD[C\vlse D]$, obtaining formulas $K_A$, $K_B$, $K_C$, and $K_D$ and derivations
  \begin{equation}\label{eq:qUelimAF2}
    \hfill
    \begin{array}{c}
      \odv{K_A\vlse K_B}{\dD[l]}{K_{l}}{\BV}
      \qomma
      \odv{K_C\vlse K_D}{\dD[r]}{K_r}{\BV}
      \qomma
      \odr{\dD[A]}{A\vlpa K_A}{\BV}
      \qomma
      \odr{\dD[B]}{B\vlpa K_B}{\BV}
      \qomma
      \odr{\dD[C]}{C\vlpa K_C}{\BV}
      \quand
      \odr{\dD[D]}{D\vlpa K_D}{\BV}
      \;\; 
      \text{such that}
    \\
      \begin{AF}
        \naf{0,0}{\nafdiagram{\AFof{\dD[A\vlse B]}}{}{,.,}21}
      \end{AF}
      =
      \AFof{\oD[A\vlse B]}
      =
      \AFof{
        \odN{\oD[A] \vlse \oD[B]}{\qdr}{(A\vlse B)\vlpa \oD[l]}{}
      }
      =
      \begin{AF}
        \naf{0,2}{\nafdiagramd{\AFof{\dD[A]}}{}{}21{}{}}
        \naf{5,2}{\nafdiagram{\AFof{\dD[B]}}{}{}21}
        \naf{5,-2}{\nafdiagramd{\AFof{\dD[l]}}{}{,.,}21{}{}}
        \naf{6,0}{\nafV1}
        \naf{-1,-1}{\nafV{2}}
        \naf{2.5,0}{\nafEXl{1.5}1}
        \naf{1,-2}{\nafV{1}}
      \end{AF}
    \\
      \begin{AF}
        \naf{0,0}{\nafdiagram{\AFof{\dD[C\vlse D]}}{}{,.,}21}
      \end{AF}
      =
      \AFof{\oD[C\vlse D]}
      =
      \AFof{
        \odN{\oD[C] \vlse \oD[D]}{\qdr}{(C\vlse D)\vlpa \oD[r]}{}
      }
      =
      \begin{AF}
        \naf{0,2}{\nafdiagramd{\AFof{\dD[C]}}{}{}21{}{}}
        \naf{5,2}{\nafdiagram{\AFof{\dD[D]}}{}{}21}
        \naf{5,-2}{\nafdiagramd{\AFof{\dD[r]}}{}{,.,}21{}{}}
        \naf{6,0}{\nafV1}
        \naf{-1,-1}{\nafV{2}}
        \naf{2.5,0}{\nafjr{1.5}1}
        \naf{2.5,0}{\nafEXl{1.5}1}
        \naf{1,-2}{\nafV{1}}
      \end{AF}
    \end{array}
  \end{equation}
  We can now build the derivation $\dDast$ as shown below, whose atomic flow below is such that $\AFof{\dDast}=\AFof{\dD}$ by the equations on atomic flows in \Cref{eq:qUelimAF1,eq:qUelimAF2} and \Cref{fig:af}.
  $$
  \adjustbox{max width=\textwidth}{$
    \odv{
      \odN{
        \odNP{\oD[A]\vlte \oD[C]}{\swir}{(A\vlte C)\vlpa K_A\vlpa K_C }{}
        \vlse
        \odNP{\oD[B]\vlte \oD[D]}{\swir}{(B\vlte D)\vlpa K_B\vlpa K_D }{}
      }{\qdr}{
        \left((A\vlte C)\vlse(B\vlte D)\right) \vlpa \odP{\odo{\odh{\oD[l]\vlpa \oD[r]}}{}{\oD[K]}{}}
      }{}
    }{\dDp[{\cC[(A\vlte C)\vlse(B\vlte D)]}]}{
      \cC[(A\vlte C)\vlse(B\vlte D)]
    }{}
  \qquad
    \begin{AF}
      \naf{0,2}{\nafdiagramd{\AFof{\dD[A]}}{}{}21{}{}}
      \naf{5,2}{\nafdiagramd{\AFof{\dD[C]}}{}{}21{}{}}
      \naf{10,2}{\nafdiagramd{\AFof{\dD[B]}}{}{}21{}{}}
      \naf{15,2}{\nafdiagram{\AFof{\dD[D]}}{}{}21}
      \naf{5,-1}{\nafJl42}
      \naf{10,-1}{\nafJl42}
      \naf{7,-1}{\nafJr{2}2}
      \naf{10.5,-1}{\nafJr{3.5}2}
      \naf{11.5,-1}{\nafV2}
      \naf{16,-1}{\nafV{2}}
      \naf{-1.5,-3.5}{\nafV{4.5}}
      \naf{4,-3.5}{\nafV{4.5}}
      \naf{5,-5.5}{\nafV{2.5}}
      \naf{7,-5.5}{\nafV{2.5}}
      \naf{10,-4}{\nafdiagramd{\AFof{\dD[l]}}{}{,.,}21{}{}}
      \naf{15,-4}{\nafdiagramd{\AFof{\dD[r]}}{}{,.,}21{}{}}
      \naf{12.5,-6.5}{\nafdiagram{\AFof{\dD[K]}}{}{,.,}{4.5}1}
      \naf{7.5,-9}{\nafdiagram{\AFof{\dD[(A\vlse B)\vlte(C\vlse D)]}}{}{,.,}{9.5}1}
    \end{AF}
  $}
  $$
\end{proof}

\section{Definition of \bvcat}\label{app:BVcat}

In this appendix we give the definition of a \bvcat with negation based on the
approach taken in \cite{blu:pan:slav:deep} with some small adaptations related to presentation.
The presentation and source code for most of the
diagrams is taken from the appendix in \cite{qcs} (arxiv version) with some
adaptations on our end.

\begin{definition}[Normal Duoidal Structure]
  Let $\mathbf{C}$ be a category with a monoidal structure
  $(\mathbf{C}, I, \otimes)$ and another monoidal structure $(\mathbf{C}, J, \vlse)$.
  We say that
  $(J, \vlse)$ is \emph{normal duoidal} to $(I, \otimes)$ if
  there exists a natural transformation:
  \begin{equation*}
    w: (A\vlse B)\otimes (C\vlse D) \to
    (A \otimes C) \vlse (B\otimes D)
  \end{equation*}
  called \emph{weak interchange}, together with morphisms
  \begin{equation*}
    i:I \cong J \qquad \delta_I : I \to I \vlse I \qquad \mu_J : J \otimes J \to J
  \end{equation*}
  which satisfy the following coherence properties:

  the weak interchange respects associativity in the sense that the following diagrams
  \begin{equation}
    \label{eq:assoc-w-seq}
    \begin{tikzcd}
      ((A\vlse B) \otimes (C \vlse D)) \otimes (E \vlse F) \arrow[r, "\alpha"] \arrow[d, "w \otimes \mathrm{id}"']             & (A\vlse B) \otimes ((C \vlse D) \otimes (E \vlse F)) \arrow[d, "\mathrm{id} \otimes w"]      \\
      ((A \otimes C) \vlse (B\otimes D)) \otimes (E \vlse F) \arrow[d, "w"']                                 & (A\vlse B) \otimes ((C \otimes E) \vlse (D \otimes F)) \arrow[d, "w"] \\
      ((A \otimes C) \otimes E) \vlse ((B \otimes D) \otimes F)  \arrow[r, "\alpha \vlse \alpha"'] & (A \otimes (C \otimes E)) \vlse (B \otimes (D \otimes F))
      \end{tikzcd}
  \end{equation}
  \bigskip

  \begin{equation}
    \label{eq:assoc-w-tens}
    \begin{tikzcd}
      ((A \vlse C) \vlse E) \otimes ((B \vlse D) \vlse F)  \arrow[r, "\alpha \otimes \alpha"] \arrow[d, "w"']   & (A \vlse (C \vlse E)) \otimes (B \vlse (D \vlse F)) \arrow[d, "w"]  \\
      ((A \vlse C) \otimes (B \vlse D)) \vlse (E \otimes F) \arrow[d, "w \vlse \mathrm{id}"']   & (A\otimes B) \vlse ((C \vlse E) \otimes (D \vlse F)) \arrow[d, "\mathrm{id} \vlse w"] \\
      ((A\otimes B) \vlse (C \otimes D)) \vlse (E \otimes F) \arrow[r, "\alpha"']           & (A\otimes B) \vlse ((C \otimes D) \vlse (E \otimes F))     \\
    \end{tikzcd}
  \end{equation}

  commute; the weak interchange respects unitality in the sense that the following diagrams

  \begin{equation}
    \label{eq:unital-tens}
    \begin{tikzcd}
      I \otimes (A \vlse B) \arrow[r, "\delta_I \otimes \text{id} "]  & (I \vlse I) \otimes (A \vlse B) \arrow[d, "w"] \\
      A \vlse B \arrow[u, "\lambda_{A\vlse B}"] \arrow[r, "\lambda_A \vlse \lambda_B"'] & (I \otimes A) \vlse (I \otimes B)
    \end{tikzcd}
    \qquad
    \begin{tikzcd}
     (A \vlse B) \otimes I \arrow[r, "\text{id} \otimes \delta_I "]  &  (A \vlse B) \otimes (I \vlse I) \arrow[d, "w"] \\
      A \vlse B \arrow[u, "\rho_{A\vlse B}"] \arrow[r, "\rho_A \vlse \rho_B"'] & (A \otimes I) \vlse (B \otimes I)
      \end{tikzcd}
  \end{equation}
  \bigskip

  \begin{equation}
    \label{eq:unital-seq}
    \begin{tikzcd}
      J \vlse (A \otimes B) & \arrow[l, "\mu_J \vlse \text{id}"']  (J \otimes J) \vlse (A \otimes B) \\
      A \otimes B \arrow[u, "\lambda_{A\otimes B}"] \arrow[r, "\lambda_A \otimes \lambda_B"'] & (J \vlse A) \otimes (J \vlse B) \arrow[u, "w"']
    \end{tikzcd}
    \qquad
    \begin{tikzcd}
      (A \otimes B) \vlse J & \arrow[l, "\text{id} \vlse \mu_J"']  (A \otimes B) \vlse (J \otimes J) \\
      A \otimes B \arrow[u, "\rho_{A\otimes B}"] \arrow[r, "\rho_A \otimes \rho_B"'] & (A \vlse J) \otimes (B \vlse J) \arrow[u, "w"']
    \end{tikzcd}
  \end{equation}
  commute;  we furthermore require that the isomorphism $i$ is an isomix map, i.e. the following diagram
  \[
    \begin{tikzcd}
      J \ltens J \arrow[d, "i^{-1} \ltens \id"] \arrow[r, "\id \otimes i^{-1}"] 
    & 
      J \ltens I \arrow[d, "\cong"]
    \\
      I \ltens J \arrow[r, "\cong"] 
    & 
      J
    \end{tikzcd}
  \]
  commutes;
  we also require that $(I, i, \delta_I)$ is a comonoid object with respect to the $(\mathcal{C}, J, \lseq)$ monoidal structure
  and we require that $(J, i^{-1}, \mu_J)$ is a monoid object with respect to the $(\mathcal{C}, I, \otimes)$ monoidal structure.
\end{definition}

\begin{definition}
  A \emph{\bvcat with negation} is a $*$-autonomous category
  ${(\mathbf{C}, I, \otimes, \multimap)}$ with an additional monoidal
  structure
  $(\mathbf{C}, J, \vlse)$, such that
  \begin{itemize}
    \item $(\mathbf{C}, J, \vlse)$ is \emph{normal duoidal} to $(\mathbf{C}, I, \otimes)$ and we write
      \[ w \colon (A\vlse B)\otimes (C\vlse D) \to   (A \otimes C) \vlse (B\otimes D) \]
      for the weak interchange natural transformation; we also require that
      $w$ commutes with the symmetry of $\otimes$ in the sense that the following diagram commutes
  \begin{equation}
    \label{def:w-comp-sym}
    \begin{tikzcd}
      (A\vlse B)\otimes (C\vlse D) \arrow[d, "w"'] \arrow[r, "\cong"] & (C\vlse D) \otimes (A\vlse B) \arrow[d, "w"]\\
      (A \otimes C) \vlse (B\otimes D) \arrow[r, "\cong"'] & (C \otimes A) \vlse (D \otimes B)
    \end{tikzcd}
  \end{equation}
    \item $(\mathbf{C}, \perp, \lpar)$ is \emph{normal duoidal} to $(\mathbf{C}, J, \vlse)$ and we write
      \[ w' \colon (A \lpar B) \vlse (C \lpar D) \to (A \vlse C) \lpar (B \vlse D) \]
      for the weak interchange natural transformation; we also require that
      $w'$ commutes with the symmetry of $\lpar$ in the sense that the following diagram commutes
  \begin{equation}
    \label{def:w-comp-sym-par}
    \begin{tikzcd}
      (A\lpar B)\vlse (C\lpar D) \arrow[d, "w'"'] \arrow[r, "\cong"] & (B \lpar A) \vlse (D \lpar C) \arrow[d, "w'"]\\
      (A \vlse C) \lpar (B \vlse D) \arrow[r, "\cong"'] & (B \vlse D) \lpar (A \vlse C)
    \end{tikzcd}
  \end{equation}
  \item We also require that the following diagram
  \begin{equation}
    \begin{tikzcd}
      (A \lseq B) \ltens ((C \lpar E) \lseq (D \lpar F))  \arrow[r, "w"] \arrow[d, "\id \otimes w'"']    & (A \ltens (C \lpar E) ) \lseq (B \ltens (D \lpar F)) \arrow[d, "a \lseq a"]      \\
      (A \lseq B) \ltens ((C \lseq D) \lpar (E \lseq F)) \arrow[d, "a"']                                 & ((A \ltens C) \lpar E) \lseq ((B \ltens D) \lpar F) \arrow[d, "w'"] \\
      ((A \lseq B) \ltens (C \lseq D)) \lpar (E \lseq F)  \arrow[r, "w \lpar \id"'] & ((A \ltens C) \lseq (B \ltens D)) \lpar (E \lseq F)
      \end{tikzcd}
  \end{equation}
  and its symmetric dual commute. Here we write
  \[ a_{ABC} \colon A \ltens (B \lpar C) \to  (A \ltens B) \lpar C  \]
  for the canonical morphism that may be defined via the $*$-autonomous structure.
  \end{itemize}
\end{definition}

This definition implies that the $*$-autonomous structure of $\CC$ is isomix.

\section{Proof of Lemma \ref{lem:interpretations}}\label{app:lem21}

\setcounter{theorem}{20}
\begin{lemma}
  For a 
  $\SBV$-derivation
  $\dD=\nosmash{\odv{A}{\dD}{B}{\SBV}}$, the following diagram commutes (in $\DD$).
  \begin{equation}
    \hfill
    \begin{tikzcd}
      U\sem A \arrow[d, "U\sem{\Phi}"'] \arrow[equal, r] 
    & 
      \afsem{A} \arrow[d, "\afsem{\Phi}"]
    \\
      U\sem{B} \arrow[equal, r] 
    & 
      \afsem{B}
    \end{tikzcd}
    \hfill
  \end{equation}
\end{lemma}
\begin{proof}
  This follows by induction on $\Phi$. The two most interesting cases are $\qdr$ and $\qur$. The case for $\qdr$ follows immediately from the definition of a strong \bvcat.
  The case for $\qur$ follows easily, because $U \sem{\qur} = U(\cong) \circ U(\sem{q}^\perp) \circ U(\cong) = \id \circ (\id \boxtimes \swap^\perp \boxtimes \id) \circ \id = \id \boxtimes \swap \boxtimes \id
  = \afsem{\qur}$
  where we used the fact that $U$ strictly preserves all the relevant structure and maps the two unnamed isomorphisms to identities in $\DD$.
  
  The remaining cases are straightforward and follow easily from standard results about $*$-autonomous and compact closed categories. We illustrate this by considering the
  $\aiur$ rule.

  In order to prove this, recall that a compact closed category $\DD$ is symmetric monoidal closed with internal hom $A \multimap B = B \otimes \lneg{A}$ and evaluation morphism
  $ \mathrm{ev}'_{A,B} \colon (A \multimap B) \otimes A \to B $
  defined by
  \[ 
    \hfill
    \mathrm{ev}'_{A,B} = \left( (A \multimap B) \otimes A = B \otimes \lneg A \otimes A \xrightarrow{\id \otimes \sigma} B \otimes A \otimes \lneg A \xrightarrow{\id \otimes \epsilon_{A}} B \right) 
    \qquad, 
    \hfill
  \]
  where $\epsilon_A \colon A \otimes \lneg A \to \lunit$ is the counit of the compact closed structure.
  By definition, we have that
  \[
    \hfill
      \sem{\nodn{a\ltens \lneg a}{\aiur}{\lunit}{}} = \sem a \otimes \lneg{\sem a} \xrightarrow{\swap} \lneg{\sem a} \otimes \sem a \xrightarrow{\mathrm{ev}} \lunit 
    \qquad.
    \hfill
  \] 
  Applying $U$ to this definition, we get
  \begin{align*}
    U(\mathrm{ev}) \circ U(\sigma) 
    &= \mathrm{ev}' \circ \sigma & \text{($U$ strictly preserves $*$-autonomous structure)} \\
    &= (\id_{\lunit} \otimes \epsilon) \circ (\id_{\lunit} \otimes \sigma) \circ \sigma  & \text{(Definition)} \\
    &= \epsilon \circ \sigma \circ \sigma & \text{($\DD$ is strict monoidal)}  \\
    &= \epsilon & \text{(Property of symmetry)} \\
    &= \Afsem{\nodn{a\ltens \lneg a}{\aiur}{\lunit}{}} & \text{(Definition)}
  \end{align*}
\end{proof}

\section{Proof of Theorem \ref{thm:caus}}\label{app:thmCaus}

\thmCaus*
\begin{proof}
  The category $\mathbf{FHilb_{sk}}$ is compact closed, because it is categorically equivalent to $\mathbf{Mat}_{\mathbb C}$ from \cite{cqm}
  which is shown to be compact closed in \cite{cqm}. Moreover, it is easy to see that it is \emph{strict} compact closed by examining the
  relevant data in \cite{cqm}.
  Furthermore, it is straightforward to see that it is \emph{dagger} compact closed (with dagger given by the conjugate transpose of matrices)
  because it is the skeletal subcategory of $\mathbf{FHilb}$
  which is well-known to be dagger compact closed \cite{cqm}.
  The $\CPs[-]$ construction (see \cite[Chapter 7]{cqm}) preserves dagger compact closure \cite[Proposition 7.26]{cqm}, so it follows that
  $\CPs[\mathbf{FHilb_{sk}}]$ is also dagger compact closed. Furthermore, $\CPs[\mathbf{FHilb_{sk}}]$ is a strict symmetric monoidal category
  \cite[Proposition 7.22]{cqm}. The double dual isomorphism  in  $\CPs[\mathbf{FHilb_{sk}}]$ is the identity, which can be easily proven using \cite[Proposition 7.24]{cqm}.
  Moreover, we also have that $(A \boxtimes B)^\perp = A^\perp \boxtimes B^\perp$ in $\CPs[\mathbf{FHilb_{sk}}]$ which can be seen by using \cite[Proposition 7.24]{cqm} and \cite[Proposition 7.22]{cqm}.
  It follows that $\CPs[\mathbf{FHilb_{sk}}]$ is a strict compact closed category. Note that since $\mathbf{FHilb_{sk}} \simeq \mathbf{FHilb}$ are categorically equivalent (the former is the skeleton of the latter),
  the results from \cite{sim:kiss:BV} may be readily reused for the category $\mathbf{FHilb_{sk}}$ as well. It follows that $\Caus[\CPs[\mathbf{FHilb_{sk}}]]$ is a \bvcat with negation (in the original sense)
  which gives us Definition \ref{def:bv-model} (1.) -- (4.). Note that the authors in \cite{sim:kiss:BV}
  have proven that the $\kappa$ isomorphism in this category is the identity, i.e. $A^\perp \lseq B^\perp = (A \lseq B)^\perp$.
  Examining the constructions in \cite{sim:kiss:BV} we see that the forgetful functor $U \colon \Caus[\CPs[\mathbf{FHilb_{sk}}]] \to \CPs[\mathbf{FHilb_{sk}}]$ preserves all the required structure strictly
  and satisfies Definition \ref{def:bv-model} (5.) thus giving us a strong \bvcat.
\end{proof}

\end{document}

%% file: macros.tex
\def\set#1{\{#1\}}
\def\Set#1{\left\{#1\right\}}

\def\tuple#1{\langle#1\rangle}
\usepackage{scalerel}

\def\intset#1#2{\set{#1,\ldots,#2}}

\renewcommand{\emptyset}{\varnothing}

\def\defn#1{{\itshape\bfseries\boldmath #1}}

\def\resp#1{(resp.~{#1})\xspace}

\mathchardef\mhyphen="2D

\newcommand\qqquad{\qquad\quad}
\newcommand{\quand}{\quad\mbox{and}\quad}
\newcommand{\qquand}{\qquad\mbox{and}\qquad}
\newcommand{\qqquand}{\qqquad\mbox{and}\qqquad}

\newcommand{\qomma}{\;,\;}

\def\bvcat{$\BV$-category\xspace}
\def\bvcats{$\BV$-categories\xspace}

\newcommand{\atomset}{\mathcal A}

\def\feq{\equiv}

\def\preq{\sim}

\newcommand{\ctx}[1][\chole]{[#1]}
\newcommand{\Ctx}[1][]{\left[{#1}\right]}

\newcommand{\conhole}{\ctx}

\newcommand{\cC}[1][\chole]{\mathcal C\!\Ctx[#1]}

\newcommand{\proves}[1][]{\mathord{\vdash_{#1}\,}}

\newcommand{\Atms}{\mathcal{A}}
\newcommand{\Fms}{\mathcal{F}}
\newcommand{\Deris}{\mathcal{D}}
\newcommand{\AFTypes}{\mathcal{Q}}

\newcommand{\dD}[1][]{\Phi_{#1}}
\newcommand{\dDb}[1][]{\Psi_{#1}}
\newcommand{\dDc}[1][]{\Theta_{#1}}
\newcommand{\dDp}[1][]{\Phi'_{#1}}
\newcommand{\dDs}[1][]{\Phi''_{#1}}
\newcommand{\dDast}[1][]{\Phi^\ast_{#1}}
\newcommand{\oD}[1][]{\odf{\dD[{#1}]}}
\newcommand{\oDp}[1][]{\odf{\dDp[{#1}]}}

\newcommandx{\genoD}[2][1=,2=]{\odv{A_{#1}}{\Phi_{#1}}{B_{#1}}{#2}}
\newcommand{\noD}[1][]{\nodf{\Phi_{#1}}}
\newcommand{\noDb}[1][]{\nodf{\Psi_{#1}}}

\def\MLL{\mathsf{MLL}}

\def\BV{\mathsf{BV}}

\def\SBV{\mathsf{SBV}}

\def\airule{\mathsf{ai}}
\def\irule{\mathsf{i}}
\def\swir{\mathsf{s}}

\def\qrule{\mathsf{q}}
\def\rrule{\mathsf{r}}

\def\idr{\irule\mathord{\downarrow}}
\def\iur{\irule\mathord{\uparrow}}
\def\aidr{\airule\mathord{\downarrow}}
\def\aiur{\airule\mathord{\uparrow}}
\def\qdr{\qrule\mathord{\downarrow}}
\def\qur{\qrule\mathord{\uparrow}}

\def\aiD{\airule\mathord{\downarrow}}
\def\aiU{\airule\mathord{\uparrow}}

\def\rR{\rrule}
\def\rU{\rrule\mathord{\uparrow}}
\def\rD{\rrule\mathord{\downarrow}}

\newcommand{\prem}[1]{\mathsf{pr}(#1)}
\newcommand{\conc}[1]{\mathsf{cn}(#1)}
\newcommand{\premf}{\mathsf{pr}}
\newcommand{\concf}{\mathsf{cn}}

\newcommand{\nosmash}[1]{#1}

\newcommand{\nodn}[4]{\vlinf{#2}{#4}{#3}{#1}}
\newcommand{\nodt}[4]{\vlidf{#2}{#4}{#3}{#1}}

\def\nodf#1{{#1}}
\def\odf#1{\od{\odh{#1}}}
\def\odP#1{\left(\od{#1}\right)}
\def\odnP#1#2#3#4{\left(\odn{#1}{#2}{#3}{#4}\right)}
\def\odNP#1#2#3#4{\left(\odN{#1}{#2}{#3}{#4}\right)}
\def\odvP#1#2#3#4{\left(\odv{#1}{#2}{#3}{#4}\right)}

\def\odrP#1#2#3{\left(\odr{#1}{#2}{#3}\right)}

\def\sysX{\mathsf{X}}

\def\scomp{\boxtimes}

\def\hcomp{\boxtimes}
\def\vcomp{\fatsemi}

\newcommand{\idF}[1][]{\mathsf{id}_{#1}}

\newcommand{\emptystring}{\varepsilon}
\newcommand{\emptyflow}{\emptyset}
\newcommand{\stringof}[1]{\mathopen{\vert\mkern-2.5mu\vert}#1\mathclose{\vert\mkern-2.5mu\vert}}
\newcommand{\stringcat}{\scomp}

\def\diaxoff{.2}
\def\diayoff{.25}

\newcommand{\MAF}[1]{\begin{tikzpicture}[nafC]#1
\end{tikzpicture}}

\newenvironment{AF}{
  \begin{tikzpicture}[nafC]
}{
  \end{tikzpicture}
}

\tikzstyle{nafZ} = [x=\nafxunit,y=\nafyunit,font=\vphantom{Ag},anchor=center,baseline={(0,0)}]
\tikzstyle{dialab} = [font=\scriptsize]

\newenvironment{AFz}{
  \begin{tikzpicture}[nafZ]
}{
  \end{tikzpicture}
}

\newcommand{\nafdia}[4][nafcolor]{
  \path[draw=#1,line width=0.01pt] (-#3-\diaxoff,-#4) -- (-#3-\diaxoff,#4) -- (#3+\diaxoff,#4) -- (#3+\diaxoff,-#4) -- cycle;
  \node[dialab] (0,0) {$\strut#2$};
}

\newcommand{\nafdiaputs}[7][nafcolor]{%
    \foreach \item [count=\i] in {#2} {\global\let\numitems\i}%
    \foreach \item [count=\i] in {#2} {%
        \ifnum\numitems=1
            \pgfmathsetmacro{\x}{#4/2}%
        \else
            \pgfmathsetmacro{\x}{-#4+#4*2*(\i-1)/(\numitems-1)}%
        \fi
        \ifnum\pdfstrcmp{\item}{.}=0
                \node at (\x,#5+\diayoff) {$\cdots$};
        \else
                \draw[draw=nafcolor,nafsingle] (\x,#5) -- (\x,#5+2*\diayoff);
                \node[above, font=\tiny] at (\x,#5) {$\item$};
        \fi
    }%
    \foreach \item [count=\j] in {#3} {\global\let\numitems\j}%
    \foreach \item [count=\j] in {#3} {%
        \ifnum\numitems=1
            \pgfmathsetmacro{\x}{#4/2}%
        \else
            \pgfmathsetmacro{\x}{-#4+#4*2*(\j-1)/(\numitems-1)}%
        \fi
        \ifnum\pdfstrcmp{\item}{.}=0
                \node at (\x,-#5-\diayoff) {$\cdots$};
        \else
                \draw[draw=nafcolor,nafsingle] (\x,-#5) -- (\x,-#5-2*\diayoff);
                \node[below, font=\tiny] at (\x,-#5-\diayoff) {$\item$};
        \fi
    }%
    \foreach \position in {#6}{
        \node at (\position,#5+2*\diayoff) {$\cdots$};
    }
    \foreach \position in {#7}{
        \node at (\position,-#5-2*\diayoff) {$\cdots$};
    }
}

\newcommand{\nafdiagram}[6][nafcolor]{%
  \naf{0,0}{\nafdia[#1]{#2}{#5}{#6}}%
  \nafdiaputs[#1]{#3}{#4}{#5}{#6}{}{}%
}

\newcommand{\nafdiagramd}[8][nafcolor]{%
  \naf{0,0}{\nafdia[#1]{#2}{#5}{#6}}%
  \nafdiaputs[#1]{#3}{#4}{#5}{#6}{#7}{#8}%
}

\newcommand{\nafcenterupshort}[3][nafcolor]{
  \draw [nafsingle,#1] (0,0) -- (0,1);
  \naflabel[#1]{-\nafcenterlabelx,.5*\nafcenterlabely}{#2}
  \naflabel[#1]{\nafcenterlabelx,.5*\nafcenterlabely}{#3}
}

\newcommand{\nafcenterdownshort}[3][nafcolor]{
  \draw [nafsingle,#1] (0,0) -- (0,-1);
  \naflabel[#1]{-\nafcenterlabelx,.5*-\nafcenterlabely}{#2}
  \naflabel[#1]{\nafcenterlabelx,.5*-\nafcenterlabely}{#3}
}

\newcommand{\nafhdshort}[5][1]{%
  \naf{-#1,0}{\nafcenterdownshort[nafcolor]{#2}{#3}}
  \naf{#1,0}{\nafcenterdownshort[nafcolor]{#4}{#5}}
  \pgfmathsetmacro{\xleft}{#1 + 0.2}
  \pgfmathsetmacro{\xright}{-#1 - 0.2}
  \path [draw=nafcolor,nafhill] (\xright,-.1) to[bend left=10] (\xleft,-.1);
}

\newcommand{\nafhushort}[5][1]{%
  \naf{-#1,0}{\nafcenterupshort[nafcolor]{#2}{#3}}
  \naf{#1,0}{\nafcenterupshort[nafcolor]{#4}{#5}}
  \pgfmathsetmacro{\xleft}{#1 + 0.2}
  \pgfmathsetmacro{\xright}{-#1 - 0.2}
  \path [draw=nafcolor,nafhill] (\xright,.1) to[bend right=10] (\xleft,.1);
}

\newcommand{\nafcapshort}[2][nafcolor]{%
  \draw[nafsingle,#1] (-#2,0) arc (180:0:#2);
}

\newcommand{\nafcupshort}[2][nafcolor]{%
  \draw[nafsingle,#1] (-#2,0) arc (180:360:#2);
}

%% file: biblio.bib
@article{category-os,
  author       = {Bert Lindenhovius and
                  Vladimir Zamdzhiev},
  title        = {The Category of Operator Spaces and Complete Contractions},
  journal      = {CoRR},
  volume       = {abs/2412.20999},
  year         = {2024},
  doi          = {10.48550/ARXIV.2412.20999},
  eprinttype    = {arXiv},
  eprint       = {2412.20999},
  bibsource    = {dblp computer science bibliography, https://dblp.org}
}

@article{star-aut-functor,
  title={Coherence of the double involution on *-autonomous categories},
  author={Cockett, J.R.B. and Hasegawa, M. and Seely, R.A.G.},
  journal={Theory and Applications of Categories},
  volume={17},
  number={2},
  pages={17--29},
  year={2006}
}

@book{mac-categories,
  title={Categories for the working mathematician},
  author={Saunders Mac Lane},
  year={1978},
  doi={10.1007/978-1-4757-4721-8},
  publisher={Springer New York, NY}
}

@article{pcbv-full-abstraction,
  author       = {Thomas Ehrhard and
                  Christine Tasson},
  title        = {Probabilistic call by push value},
  journal      = {Log. Methods Comput. Sci.},
  volume       = {15},
  number       = {1},
  year         = {2019},
  doi          = {10.23638/LMCS-15(1:3)2019},
  bibsource    = {dblp computer science bibliography, https://dblp.org}
}

@article{ppcf-full-abstraction,
  author       = {Thomas Ehrhard and
                  Michele Pagani and
                  Christine Tasson},
  title        = {Full Abstraction for Probabilistic {PCF}},
  journal      = {J. {ACM}},
  volume       = {65},
  number       = {4},
  pages        = {23:1--23:44},
  year         = {2018},
  doi          = {10.1145/3164540},
  bibsource    = {dblp computer science bibliography, https://dblp.org}
}

@inproceedings{os-lics,
  author       = {Bert Lindenhovius and
                  Vladimir Zamdzhiev},
  title        = {Operator Spaces, Linear Logic and the {H}eisenberg-{S}chr{\"{o}}dinger
                  Duality of Quantum Theory},
  booktitle    = {40th Annual {ACM/IEEE} Symposium on Logic in Computer Science, {LICS}
                  2025, Singapore, June 23-26, 2025},
  pages        = {870--883},
  publisher    = {{IEEE}},
  year         = {2025},
  url          = {https://doi.org/10.1109/LICS65433.2025.00071},
  doi          = {10.1109/LICS65433.2025.00071},
  bibsource    = {dblp computer science bibliography, https://dblp.org}
}

@article{HS,
title = {Glueing and orthogonality for models of linear logic},
journal = {Theoretical Computer Science},
volume = {294},
number = {1},
pages = {183-231},
year = {2003},
issn = {0304-3975},
doi = {https://doi.org/10.1016/S0304-3975(01)00241-9},
author = {Martin Hyland and Andrea Schalk},
}

@article{carbone:99,
  author    = {Alessandra Carbone},
  title     = {Turning Cycles into Spirals},
  journal   = {Ann. Pure Appl. Logic},
  volume    = {96},
  number    = {1-3},
  pages     = {57--73},
  year      = {1999},
}

@book{blecher-merdy,
    author = {Blecher, David P. and Le Merdy, Christian},
    title = "{Operator Algebras and Their Modules: An operator space approach}",
    publisher = {Oxford University Press},
    year = {2004},
    month = {10},
    isbn = {9780198526599},
    doi = {10.1093/acprof:oso/9780198526599.001.0001},
}

@book{er2000operator,
  title={Operator Spaces},
  author={Effros, E.G. and Ruan, Z.J.},
  isbn={9780198534822},
  lccn={lc00056685},
  series={London Mathematical Society monographs},
  url={https://books.google.fr/books?id=v7mj8Dy84k8C},
  year={2000},
  publisher={Clarendon Press}
}

@book{Pisier_2003,
  place={Cambridge},
  series={London Mathematical Society Lecture Note Series},
  title={Introduction to Operator Space Theory},
  publisher={Cambridge University Press},
  author={Pisier, Gilles},
  year={2003},
  collection={London Mathematical Society Lecture Note Series}
}

@book{cqm,
    author = {Heunen, Chris and Vicary, Jamie},
    title = {Categories for Quantum Theory: An Introduction},
    publisher = {Oxford University Press},
    year = {2019},
    month = {11},
    isbn = {9780198739623},
    doi = {10.1093/oso/9780198739623.001.0001},
    url = {https://doi.org/10.1093/oso/9780198739623.001.0001},
    eprint = {https://academic.oup.com/book/43710/book-pdf/50991591/9780191060069_web.pdf},
}

@article{bv-cats2,
  author       = {James Hefford and
                  Matthew Wilson},
  title        = {A {BV}-Category of Spacetime Interventions},
  journal      = {CoRR},
  volume       = {abs/2502.19022},
  year         = {2025},
  url          = {https://doi.org/10.48550/arXiv.2502.19022},
  doi          = {10.48550/ARXIV.2502.19022},
  eprinttype    = {arXiv},
  eprint       = {2502.19022},
  bibsource    = {dblp computer science bibliography, https://dblp.org}
}

@article{er-shuffle,
 ISSN = {03794024, 18417744},
 URL = {http://www.jstor.org/stable/24718935},
 author = {Edward G. Effros and Zhong-Jin Ruan},
 journal = {Journal of Operator Theory},
 number = {1},
 pages = {131--156},
 publisher = {Theta Foundation},
 title = {OPERATOR SPACE TENSOR PRODUCTS AND {H}OPF CONVOLUTION ALGEBRAS},
 urldate = {2025-09-04},
 volume = {50},
 year = {2003}
}

@InProceedings{gug:gun:par:RTA10,
  author =	{Guglielmi, Alessio and Gundersen, Tom and Parigot, Michel},
  title =	{{A Proof Calculus Which Reduces Syntactic Bureaucracy}},
  booktitle =	{Proceedings of the 21st International Conference on Rewriting Techniques and Applications},
  pages =	{135--150},
  series =	{Leibniz International Proceedings in Informatics (LIPIcs)},
  ISBN =	{978-3-939897-18-7},
  ISSN =	{1868-8969},
  year =	{2010},
  volume =	{6},
  editor =	{Lynch, Christopher},
  publisher =	{Schloss Dagstuhl -- Leibniz-Zentrum f{\"u}r Informatik},
  address =	{Dagstuhl, Germany},
  URN =		{urn:nbn:de:0030-drops-26490},
  doi =		{10.4230/LIPIcs.RTA.2010.135}
}

@misc{acc:str:IBVext,
	title={Intuitionistic {BV} (Extended version)}, 
	author={Matteo Acclavio and Lutz Strassburger},
	year={2025},
	eprint={2505.13284},
	archivePrefix={arXiv},
	primaryClass={cs.LO},
	url={https://arxiv.org/abs/2505.13284}, 
}

@InProceedings{acc:str:IBV,
	author="Acclavio, Matteo
	and Stra{\ss}burger, Lutz",
	editor="Pozzato, Gian Luca
	and Uustalu, Tarmo",
	title="Intuitionistic {BV}",
	booktitle="Automated Reasoning with Analytic Tableaux and Related Methods",
	year="2026",
	publisher="Springer Nature Switzerland",
	address="Cham",
	pages="414--432",
	isbn="978-3-032-06085-3"
}

@article{SIS-V,
  author       = {Alessio Guglielmi and
                  Lutz Stra{\ss}burger},
  title        = {A system of interaction and structure {V:} the exponentials and splitting},
  journal      = {Math. Struct. Comput. Sci.},
  volume       = {21},
  number       = {3},
  pages        = {563--584},
  year         = {2011},
}

@InProceedings{sim:kiss:BV,
  author =	{Simmons, Will and Kissinger, Aleks},
  title =	{Higher-Order Causal Theories Are Models of {BV}-Logic},
  booktitle =	{47th International Symposium on Mathematical Foundations of Computer Science (MFCS 2022)},
  pages =	{80:1--80:14},
  series =	{Leibniz International Proceedings in Informatics (LIPIcs)},
  ISBN =	{978-3-95977-256-3},
  ISSN =	{1868-8969},
  year =	{2022},
  volume =	{241},
  editor =	{Szeider, Stefan and Ganian, Robert and Silva, Alexandra},
  publisher =	{Schloss Dagstuhl -- Leibniz-Zentrum f{\"u}r Informatik},
  address =	{Dagstuhl, Germany},
  doi =		{10.4230/LIPIcs.MFCS.2022.80}
}

@article{dan:ehr:prob,
title = {Probabilistic coherence spaces as a model of higher-order probabilistic computation},
journal = {Information and Computation},
volume = {209},
number = {6},
pages = {966-991},
year = {2011},
issn = {0890-5401},
doi = {https://doi.org/10.1016/j.ic.2011.02.001},
url = {https://www.sciencedirect.com/science/article/pii/S0890540111000411},
author = {Vincent Danos and Thomas Ehrhard}
}

@inproceedings{girard:96:PN,
	author = {Jean-Yves Girard},
	title = {Proof-nets : the parallel syntax for proof-theory},
	booktitle = {Logic and Algebra},
	publisher = {Marcel Dekker, New York},
	editor = {Aldo Ursini and Paolo Agliano},
	year = {1996}
}

@article{lam:str:06:freestar,
   author =  {Fran\c{c}ois Lamarche and Lutz Stra{\ss}burger},
   title =   {From Proof Nets to the Free *-Autonomous Category},
   journal="Logical Methods in Computer Science",
   volume=2,
number="4:3",
   year =    2006,
pages="1--44",
url = "http://arxiv.org/abs/cs/0605054",
}

@Article{blute:93, 
  author = 	"Richard Blute", 
  title = 	"Linear Logic, Coherence and Dinaturality", 
  journal = 	tcs,
  year = 	1993,
  volume = 	115,
  pages = 	"3--41"
}

@unpublished{hughes:freestar,
author = "Dominic Hughes",
title = "Simple free star-autonomous categories and full coherence",
year = 2005,
note="Preprint", 
url={http://arxiv.org/abs/math.CT/0506521},
}

@article{hor:tiu:19, 
	title={Constructing weak simulations from linear implications for processes with private names}, 
	volume={29}, 
	DOI={10.1017/S0960129518000452}, 
	number={8}, 
	journal={Mathematical Structures in Computer Science}, 
	author={Horne, Ross and Tiu, Alwen}, 
	year={2019}, 
	pages={1275–1308}
}

@InProceedings{hor:tiu:ama:cio:private,
	author =	{Horne, Ross and Tiu, Alwen and Aman, Bogdan and Ciobanu, Gabriel},
	title =	{{Private Names in Non-Commutative Logic}},
	booktitle =	{27th International Conference on Concurrency Theory (CONCUR 2016)},
	pages =	{31:1--31:16},
	series =	{Leibniz International Proceedings in Informatics (LIPIcs)},
	ISBN =	{978-3-95977-017-0},
	ISSN =	{1868-8969},
	year =	{2016},
	volume =	{59},
	editor =	{Desharnais, Jos\'{e}e and Jagadeesan, Radha},
	publisher =	{Schloss Dagstuhl -- Leibniz-Zentrum f{\"u}r Informatik},
	address =	{Dagstuhl, Germany},
	doi =		{10.4230/LIPIcs.CONCUR.2016.31}
}

@inproceedings{bru:02,
	title={A purely logical account of sequentiality in proof search},
	author={Bruscoli, Paola},
	booktitle={International Conference on Logic Programming},
	pages={302--316},
	year={2002},
	organization={Springer}
}

@article{abramsky:jagadeesan:94,
   author =  {Samson Abramsky and Radha Jagadeesan},
   title =   {Games and Full Completeness for Multiplicative Linear Logic},
   journal = {Journal of Symbolic Logic},
   volume =  59,
   number =  2,
   pages =   {543--574},
   year =    1994,
}

@article{fle:ret:mix, 
	title={The mix rule}, 
	volume={4}, 
	DOI={10.1017/S0960129500000451}, 
	number={2}, 
	journal={Mathematical Structures in Computer Science}, 
	author={Fleury, Arnaud and Retoré, Christian}, 
	year={1994}, 
	pages={273–285}
}

@article{gir:ll,
	title = {Linear logic},
	journal = {Theoretical Computer Science},
	volume = {50},
	number = {1},
	pages = {1-101},
	year = {1987},
	issn = {0304-3975},
	doi = {10.1016/0304-3975(87)90045-4},
	author = {Jean-Yves Girard}
}

@article{lau:reg:89,
	author = {Danos, Vincent and Regnier, Laurent},
	date = {1989/10/01},
	doi = {10.1007/BF01622878},
	id = {Danos1989},
	isbn = {1432-0665},
	journal = {Archive for Mathematical Logic},
	number = {3},
	pages = {181--203},
	title = {The structure of multiplicatives},
	url = {https://doi.org/10.1007/BF01622878},
	volume = {28},
	year = {1989}
}

@article{blu:pan:slav:deep,
	author = {Blute, Richard and Panangaden, Prakash and Slavnov, Sergey},
	date = {2012/06/01},
	doi = {10.1007/s10485-010-9241-0},
	id = {Blute2012},
	isbn = {1572-9095},
	journal = {Applied Categorical Structures},
	number = {3},
	pages = {209--228},
	title = {Deep Inference and Probabilistic Coherence Spaces},
	volume = {20},
	year = {2012}
	}

@Incollection{blu:gug:iva:pan:str:quantum,
	author="Blute, Richard F.
	and Guglielmi, Alessio
	and Ivanov, Ivan T.
	and Panangaden, Prakash
	and Stra{\ss}burger, Lutz",
	editor="Casadio, Claudia
	and Coecke, Bob
	and Moortgat, Michael
	and Scott, Philip",
	title="A Logical Basis for Quantum Evolution and Entanglement",
	bookTitle="Categories and Types in Logic, Language, and Physics: Essays Dedicated to Jim Lambek on the Occasion of His 90th Birthday",
	year="2014",
	publisher="Springer Berlin Heidelberg",
	address="Berlin, Heidelberg",
	pages="90--107",
	isbn="978-3-642-54789-8",
	doi="10.1007/978-3-642-54789-8_6",
	url="https://doi.org/10.1007/978-3-642-54789-8_6"
}

@article{buss:91,
	title = {The undecidability of k-provability},
	journal = {Annals of Pure and Applied Logic},
	volume = {53},
	number = {1},
	pages = {75-102},
	year = {1991},
	issn = {0168-0072},
	doi = {https://doi.org/10.1016/0168-0072(91)90059-U},
	url = {https://www.sciencedirect.com/science/article/pii/016800729190059U},
	author = {Samuel {R. Buss}}
}

@phdthesis{retore:phd,
	author =  {Retor{\'e}, Christian},
	title =   "R{\'e}seaux et S{\'e}quents Ordonn{\'e}s",
	school =  {Universit{\'e} Paris VII},
	year =    1993,
}

@incollection{retore:99,
   author =    {Christian Retor\'e},
   title =     {Pomset Logic as a Calculus of Directed Cographs},
   editor =    {V. M. Abrusci and C. Casadio},
   booktitle = {Dynamic Perspectives in Logic and Linguistics},
   year =      1999,
   publisher = {Bulzoni},
   address =   {Roma},
   pages =     {221--247},
   note =      {Also available as INRIA Rapport de Recherche RR-3714},
}

@article{retore:97, 
	title={A semantic characterisation of the correctness of a proof net}, 
	volume={7}, 
	DOI={10.1017/S096012959700234X}, 
	number={5}, 
	journal={Mathematical Structures in Computer Science}, 
	author={Retor{\'e}, Christian}, 
	year={1997}, 
	pages={445–452}
}

@article{selinger:dagger,
title = {Dagger Compact Closed Categories and Completely Positive Maps: (Extended Abstract)},
journal = {Electronic Notes in Theoretical Computer Science},
volume = {170},
pages = {139-163},
year = {2007},
note = {Proceedings of the 3rd International Workshop on Quantum Programming Languages (QPL 2005)},
issn = {1571-0661},
doi = {https://doi.org/10.1016/j.entcs.2006.12.018},
url = {https://www.sciencedirect.com/science/article/pii/S1571066107000606},
author = {Peter Selinger}
}

@Article{kelly:maclane:71,
  author =       {Gregory Maxwell Kelly and Saunders {Mac~Lane}},
  title =        {Coherence in Closed Categories},
  journal =      "J.\ of Pure and Applied Algebra",
  year =         {1971},
  volume =       {1},
  pages =        {97--140},
}

@Article{kelly:laplaza:80,
  author =       {Gregory Maxwell Kelly and M.~L.~Laplaza},
  title =        {Coherence for compact closed categories},
  journal =      "J.\ of Pure and Applied Algebra",
  year =         {1980},
  volume =       {19},
  pages =        {193--213},
}

@article{ret:newPomset,
	title={Pomset Logic: The other approach to noncommutativity in logic},
	author={Retor{\'e}, Christian},
	journal={Joachim Lambek: The Interplay of Mathematics, Logic, and Linguistics},
	pages={299--345},
	year={2021},
	publisher={Springer}
}

@inproceedings{gug:str:01,
   author =      {Alessio Guglielmi and Lutz Stra{\ss}burger},
   title =       {Non-commutativity and {MELL} in the Calculus of Structures},
   year =        2001,
   booktitle =   {Computer Science Logic, CSL 2001},
   publisher =   {Springer-Verlag},
   volume =      2142,
   pages =       {54--68},
   editor =      {Laurent Fribourg},
   series =      {LNCS}
}

@article{gug:SIS,
author = {Guglielmi, Alessio},
title = {A system of interaction and structure},
year = {2007},
issue_date = {January 2007},
publisher = {Association for Computing Machinery},
address = {New York, NY, USA},
volume = {8},
number = {1},
issn = {1529-3785},
doi = {10.1145/1182613.1182614},
journal = {ACM Trans. Comput. Logic},
month = jan,
}

@book{polybook, 
	place={Cambridge}, 
	series={London Mathematical Society Lecture Note Series}, 
	title={Polygraphs: From Rewriting to Higher Categories}, 
	publisher={Cambridge University Press}, 
	author={Ara, Dimitri and Burroni, Albert and Guiraud, Yves and Malbos, Philippe and Métayer, François and Mimram, Samuel}, 
	year={2025}, 
	collection={London Mathematical Society Lecture Note Series}
}

@article{string:diagrams:RTI,
	author = {Bonchi, Filippo and Gadducci, Fabio and Kissinger, Aleks and Sobocinski, Pawel and Zanasi, Fabio},
	title = {String Diagram Rewrite Theory {I}: Rewriting with {F}robenius Structure},
	year = {2022},
	issue_date = {April 2022},
	publisher = {Association for Computing Machinery},
	address = {New York, NY, USA},
	volume = {69},
	number = {2},
	issn = {0004-5411},
	url = {https://doi.org/10.1145/3502719},
	doi = {10.1145/3502719},
	journal = {J. ACM},
	month = mar,
	articleno = {14},
	numpages = {58}
}

@article{lafont:boolean,
	title = {Towards an algebraic theory of Boolean circuits},
	journal = {Journal of Pure and Applied Algebra},
	volume = {184},
	number = {2},
	pages = {257-310},
	year = {2003},
	issn = {0022-4049},
	doi = {https://doi.org/10.1016/S0022-4049(03)00069-0},
	url = {https://www.sciencedirect.com/science/article/pii/S0022404903000690},
	author = {Yves Lafont}
}

@article{string:diagrams:RTII, 
	title={String diagram rewrite theory {II}: Rewriting with symmetric monoidal structure}, 
	volume={32}, 
	DOI={10.1017/S0960129522000317}, 
	number={4}, 
	journal={Mathematical Structures in Computer Science}, 
	author={Bonchi, Filippo and Gadducci, Fabio and Kissinger, Aleks and Sobocinski, Pawel and Zanasi, Fabio}, 
	year={2022}, 
	pages={511–541}
}

@phdthesis{tubella:phd,
author="Aler Tubella, Andrea",
title="A study of normalisation through subatomic logic",
year=2016,
school="University of Bath",
}

@article{SIS-IV,
	author = {Lutz Stra{\ss}burger and Guglielmi, Alessio},
	title = {A system of interaction and structure {IV}: The exponentials and decomposition},
	year = {2011},
	issue_date = {July 2011},
	publisher = {Association for Computing Machinery},
	address = {New York, NY, USA},
	volume = {12},
	number = {4},
	issn = {1529-3785},
	url = {https://doi.org/10.1145/1970398.1970399},
	doi = {10.1145/1970398.1970399},
	journal = {ACM Trans. Comput. Logic},
	month = jul,
	articleno = {23},
	numpages = {39}
}

@inproceedings{gug:gun:str:LICS10,
  author    = {Alessio Guglielmi and
               Tom Gundersen and
               Lutz Stra{\ss}burger},
  title     = {Breaking Paths in Atomic Flows for Classical Logic},
  booktitle = {Proceedings of the 25th Annual {IEEE} Symposium on Logic in Computer
               Science, {LICS} 2010, 11-14 July 2010, Edinburgh, United Kingdom},
  pages     = {284--293},
  year      = {2010},
  publisher = {{IEEE} Computer Society},
}

@article{gug:gun:flows,
    title      = {Normalisation Control in Deep Inference via Atomic Flows},
    author     = {Alessio Guglielmi and Tom Gundersen},
    url        = {https://lmcs.episciences.org/1081},
    doi        = {10.2168/LMCS-4(1:9)2008},
    journal    = {Logical Methods in Computer Science},
    issn       = {1860-5974},
    eid        = 9,
    year       = {2008},
    month      = {Mar},
numpages={36},
volume=4,
number="1:9"
}

@article{tiu:SIS-II,
	author = "Tiu, Alwen Fernanto",
	title = "A System of Interaction and Structure {II}:
	{T}he Need for Deep Inference",
	journal="Logical Methods in Computer Science",
	volume=2,
	number=2,
	pages="1--24",
	year = "2006",
	doi = "10.2168/LMCS-2(2:4)2006",
}

@article{tito:str:SIS-III,
  author       = {L{\^{e}} Th{\`{a}}nh Dung Nguy{\^{e}}n and
                  Lutz Stra{\ss}burger},
  title        = {A System of Interaction and Structure {III:} The Complexity of {BV}
                  and Pomset Logic},
  journal      = {Log. Methods Comput. Sci.},
  volume       = {19},
  number       = {4},
  year         = {2023},
  doi          = {10.46298/LMCS-19(4:25)2023},
}

@InProceedings{tito:lutz:csl22,
	author =	{Lê Thành Dũng Nguyên and Lutz Straßburger},
	title =	{{BV and Pomset Logic are not the same}},
	booktitle =	{30th EACSL Annual Conference on Computer Science Logic (CSL 2022)},
	pages =	{3:1--3:17},
	series =	{Leibniz International Proceedings in Informatics (LIPIcs)},
	ISBN =	{978-3-95977-218-1},
	ISSN =	{1868-8969},
	year =	{2022},
	volume =	{216},
	editor =	{Manea, Florin and Simpson, Alex},
	publisher =	{Schloss Dagstuhl -- Leibniz-Zentrum f{\"u}r Informatik},
	address =	{Dagstuhl, Germany},
	doi =		{10.4230/LIPIcs.CSL.2022.3},
}

@article{gir:quantum,
  title={Between Logic and Quantic: a Tract},
  author={Girard, Jean-Yves},
  journal={Linear Logic in Computer Science, London Mathematical Society Lecture Notes Series},
  volume={316},
  year={2004},
  publisher={Cambridge University Press}
}

@phdthesis{simmons:phd,
  edition = {},
  number = {},
  journal = {},
  booktitle = {},
  pages = {},
  publisher = {University of Oxford},
  school = {University of Oxford},
  title = {Causal types for higher-order quantum theory},
  volume = {},
  author = {Will Simmons},
  editor = {},
  year = {2024},
  series = {}
}

@misc{sim:kis:complete,
      title={A complete logic for causal consistency}, 
      author={Will Simmons and Aleks Kissinger},
      year={2024},
      eprint={2403.09297},
      archivePrefix={arXiv},
      primaryClass={cs.LO},
      url={https://arxiv.org/abs/2403.09297}, 
}

@inproceedings{qcs,
  author       = {Thea Li and
                  Vladimir Zamdzhiev},
  editor       = {Nathalie Bertrand and
                  Stefan Milius},
  title        = {Quantum Coherence Spaces Revisited: {A} von {N}eumann (Co)Algebraic
                  Approach},
  booktitle    = {Foundations of Software Science and Computation Structures - 29th
                  International Conference, FoSSaCS 2026, Held as Part of the International
                  Joint Conferences on Theory and Practice of Software, {ETAPS} 2026,
                  Turin, Italy, April 11-16, 2026, Proceedings},
  series       = {Lecture Notes in Computer Science},
  pages        = {418--439},
  publisher    = {Springer},
  year         = {2026},
  url          = {https://doi.org/10.1007/978-3-032-22730-0\_20},
  doi          = {10.1007/978-3-032-22730-0\_20},
  eprint       = {2601.15832},
  bibsource    = {dblp computer science bibliography, https://dblp.org}
}
